\documentclass[runningheads,orivec]{llncs}
\usepackage[T1]{fontenc}

\newif\iflong

\usepackage{xspace}
\usepackage{orcidlink}
\usepackage{amssymb}
\usepackage{mathtools}
\usepackage{nicefrac}
\usepackage{multirow,makecell}
\usepackage{multibib}

\usepackage{graphicx}
\usepackage{tikz}
\usetikzlibrary{positioning, arrows.meta, calc, decorations.pathreplacing, fit, backgrounds}
\usepackage{xcolor}

\usepackage{hyperref}
\hypersetup{
    pdfencoding=auto, 
    psdextra,
    colorlinks=true,
    citecolor=green!40!black,
    linkcolor=blue!90!black,
    urlcolor=blue!80!black
}

\usepackage{cleveref}
\usepackage{cite}
\usepackage{paralist}

\usepackage[inline]{enumitem}

\usepackage{tabularx}
\usepackage{booktabs}
\newcolumntype{Y}{>{\centering\arraybackslash}X}

\usepackage{comment}
\usepackage{todonotes}

\newcommand{\probName}[1]{\textsc{#1}\xspace}
\def\IRSCCfo{\probName{IR-FO-SCC}}
\def\IRFSCfo{\probName{IR-FO-FSC}}
\def\IRBSCfo{\probName{IR-FO-BSC}}
\def\IRkfo{\probName{IR-FO-$k$}}
\def\IRSCCeo{\probName{IR-EO-SCC}}

\def\IRBSCeo{\probName{IR-EO-BSC}}
\def\IRkeo{\probName{IR-EO-$k$}}

\newcommand{\agents}{A}
\newcommand{\numAgents}{n}

\newcommand{\friends}{\mathcal{F}}
\newcommand{\enemies}{\mathcal{E}}
\newcommand{\neutrals}{\mathcal{N}}

\newcommand{\numConflicting}{\xi}

\DeclareMathOperator{\cFn}{c}

\newcommand{\numCoals}{k}
\DeclareMathOperator{\ub}{ub}
\DeclareMathOperator{\lb}{lb}

\newcommand{\pttn}{\pi}
\newcommand{\wt}[1]{\widetilde{#1}}
\newcommand{\T}{\mathcal{T}}
\def\ind{\mathsf{ind}}

\DeclareMathOperator{\tw}{tw}
\DeclareMathOperator{\vc}{vc}
\DeclareMathOperator{\cw}{cw}

\newcommand{\true}{\texttt{true}\xspace}
\newcommand{\false}{\texttt{false}\xspace}
\newcommand{\No}{\emph{no}\xspace}
\newcommand{\Yes}{\emph{yes}\xspace}
\newcommand{\DP}{\operatorname{DP}}

\newcommand{\Oh}[1]{{\mathcal{O}\mathopen{}\left(#1\right)}}

\NewDocumentCommand{\cc}{ O{} O{} m }{\mbox{%
    \expandafter\ifx\expandafter\relax\detokenize{#2}\relax\else{#2-}\fi%
    \textsf{#3}%
    \expandafter\ifx\expandafter\relax\detokenize{#1}\relax\else{-#1}\fi%
    }\xspace}
\newcommand{\NP}{\cc{NP}}
\newcommand{\NPh}{\cc[hard]{NP}}
\newcommand{\NPhness}{\cc[hardness]{NP}}
\newcommand{\NPc}{\cc[complete]{NP}}
\newcommand{\NPcness}{\cc[completeness]{NP}}
\newcommand{\FPT}{\cc{FPT}}
\newcommand{\XP}{\cc{XP}}
\newcommand{\W}[1][1]{\cc{W[#1]}}
\newcommand{\Wh}[1][1]{\cc[hard]{W[#1]}}
\newcommand{\Whness}[1][1]{\cc[hardness]{W[#1]}}

\newcommand{\pNPh}{\cc[hard][para]{NP}}

\spnewtheorem{observation}{Observation}{\bfseries}{\itshape}
\Crefname{observation}{Observation}{Observations}

\spnewtheorem{claim}{Claim}{\bfseries}{\itshape}
\Crefname{claim}{Claim}{Claims}
\let\doendproof\endproof
\renewcommand\endproof{~\hfill$\square$\doendproof}

\usepackage{pict2e}
\newcommand{\opentriangle}{%
  \raisebox{0.2pt}{\makebox[0.77778em]{%
    \setlength{\unitlength}{0.6em}%
    \linethickness{0.4pt}%
    \begin{picture}(1,1)
    \polygon(0,0)(1,0)(1,1)
    \end{picture}%
  }}%
}
\newenvironment{proofsketch}[1]{\emph{Proof sketch.}\hspace{0.15cm}#1}{\hfill$\opentriangle$\medskip}

\newenvironment{claimproof}[1]{\emph{Proof.}\hspace{0.15cm}#1}{\hfill$\blacktriangleleft$\medskip}

\newcommand{\proofsubparagraph}[1]{\smallskip\noindent{\textbf{#1.}}\hspace{0.25cm}}

\newcommand{\proofcase}[3][Case]{\smallskip{\emph{#1 #2: #3}}\hspace{0.15cm}}

\usepackage[title]{appendix}
\usepackage{thm-restate}
\newcommand{\linkproof}[1]{%
    \hyperref[app:proof:#1]{$\star$}%
}

\newcommand{\toappendix}[1]{%
  \gappto{\appendixtext}{
    {#1}
   }
}

\newcommand{\prooftoappendix}[3]{%
  \gappto{\appendixtext}{
    \subsection{Proof of \Cref{#1}}\label{app:proof:#1}
    #2
    \begin{proof}
    #3\end{proof}
  }
}

\newcommand{\appendixsection}[1]{%
  \gappto{\appendixtext}{
    \section{Additional material for \Cref{#1}}
    \label{app:#1}
  }
}

\newcites{app}{References for the Appendices}

\definecolor{cbBlue}{HTML}{332288}
\definecolor{cbOrange}{HTML}{DDCC77}
\definecolor{cbGreen}{HTML}{44AA99}
\definecolor{cbRed}{HTML}{882255}

\begin{document}

\title{Individual Rationality in Constrained Hedonic Games: Friends, Enemies, and Neutrals}
\titlerunning{Individual Rationality in Constrained Hedonic Games}

\author{%
    Šimon Schierreich\inst{1,2}\orcidlink{0000-0001-8901-1942} \and
    Ildikó Schlotter\inst{3,4}\orcidlink{0000-0002-0114-8280}
}
\authorrunning{Šimon Schierreich and Ildikó Schlotter}
\institute{%
    AGH University of Krakow, Poland \and
    Czech Technical University in Prague, Czechia \and
    ELTE Centre for Economic and Regional Studies, Hungary \and
    Budapest University of Technology and Economics, Hungary\\
    \email{schiesim@fit.cvut.cz}, \email{schlotter.ildiko@krtk.elte.hu}%
}

\maketitle

\begin{abstract}
    We study constrained coalition formation in games induced by \emph{friends}, \emph{enemies}, and \emph{neutrals}, under the two standard refinements of additively separable preferences: \emph{friend-oriented} and \emph{enemy-oriented}. We ask for partitions that are \emph{individually rational} (IR), while additionally requiring exactly~$\numCoals$ non-empty coalitions, each satisfying a prescribed lower and upper bound on its size. Although IR alone is trivial to satisfy for any hedonic game, the size constraints make it computationally intractable to decide whether a feasible partition exists.
    
    The two models tell strikingly different stories. Under enemy-oriented preferences, the problem collapses to size-constrained graph coloring, and its complexity follows accordingly. Under friend-oriented preferences, however, the picture is far more intricate, and is governed by the \emph{enmity structure} rather than the friendships. The complexity is further shaped by two factors: how strict the imposed size requirements are, and whether relationships are \emph{symmetric} or \emph{asymmetric}, with several cases turning out tractable in the symmetric setting but intractable once asymmetry is allowed. Charting this boundary in terms of both classical and parameterized complexity, we provide a complete understanding of which properties of the friend/enemy structure are responsible for hardness.

\end{abstract}

\section{Introduction}

Imagine that you are running a summer camp and must assign the arriving children to the camp's cabins for the week. The children are far from indifferent about who they bunk with: some are close friends who would love to stay together, a few pairs genuinely dislike each other and would make life miserable for their whole cabin if put together, and most are simply indifferent to the rest. The assignment is also bound by hard practical limits. Each cabin has a maximum capacity; no cabin should be opened with only a child or two inside, and the number of cabins available for the week is fixed in advance.

This situation, and many others like it, captures the essence of the broader problem of \emph{coalition formation}. The task is to partition a set of agents into \emph{coalitions} so that the resulting configuration is \emph{stable}\footnote{Other desirable criteria are also studied; e.g., various notions of \emph{efficiency}~\cite{Bullinger2020,GanianHKRSS2023,BullingerCS2025}.}---that is, no agent would prefer to deviate from their current coalition. If we further assume that the agents are \emph{self-interested} and care only about the composition of their own coalition, we obtain the well-known model of \emph{hedonic games}~\cite{DrezeG1980,AzizS2016}. In this work, we study hedonic games induced by \emph{friends}, \emph{enemies}, and \emph{neutrals}~\cite{DimitrovBHS2006,OhtaBISY2017,ChenCRS2023,RotheSS2018}: each agent partitions the others into those they like, those they dislike, and those they are indifferent to. This class is a special case of a more general model of \emph{additively separable hedonic games} (ASHGs)~\cite{BogomolnaiaJ2002,BanerjeeKS2001}. Following the literature, we consider the two canonical refinements of such preferences, \emph{friend-oriented} and \emph{enemy-oriented}, which differ in whether an agent primarily seeks the company of their friends or primarily seeks to avoid their enemies.

The notion of stability we adopt is that of \emph{individual rationality} (IR)~\cite{Roth1977,AshlagiR2011,DeligkasEKS2024b}. %
Intuitively, IR requires that no agent prefers leaving their current coalition and forming a coalition on their own. Traditionally, any coalition structure is acceptable as long as it meets the chosen stability criterion. In many real-world scenarios, however---from assigning children to camp cabins to forming student project teams or assigning researchers to working groups---we face additional structural constraints on feasible outcomes:
\begin{enumerate}
    \item We require that the resulting coalition structure consists of exactly~$\numCoals$ non-empty coalitions. This constraint was introduced by Sless \emph{et al.}~\cite{SlessHKW2018} and subsequently studied by Waxman \emph{et al.}~\cite{WaxmanKH2021} and Barr \emph{et al.}~\cite{BarrTKRH24}.
    \item The size of each of the $k$ coalitions must be within its own lower and upper bound.
\end{enumerate}
It is easy to see that the first constraint is a special case of the second, where the lower bound of each coalition is one, and the upper bound equals the number of agents. The general model combining these constraints with IR was introduced by Fioravantes \emph{et al.}~\cite{FioravantesGMS2026a} for ASHGs, building on the fixed-size coalitions of Bil\'o \emph{et al.}~\cite{BiloMM2022} and a line of related work on fixed- and equal-size coalitions~\cite{LiMNS2023,DeligkasEIKS2025,AgarwalARN2025} (see also the recent model of Bullinger \emph{et al.}~\cite{BullingerDEG2025} with global size constraints). The particular structural restrictions we study here were subsequently investigated by Fioravantes \emph{et al.}~\cite{FioravantesGMS2026b} for anonymous and diversity preferences~\cite{BredereckEI2019,GanianHKSS2023}.%

\begin{table}[bt!]
    \caption{An overview of our complexity results for friend-oriented preferences.
    Here, 
    $k$ is the number of desired coalitions, 
    $\ub=\max_{j \in [k]} \ub(j)$, 
    $r$ denotes the number of recluses (i.e., agents~$a$ with $\friends_a=\emptyset$ and $\enemies_a \neq \emptyset$),
    $\vc$ (and $\vc^\enemies$), $\tw$, and $\cw$ the vertex cover number of the relationship graph (the enmity graph, resp.), its treewidth, and its clique-width,
    and $\numConflicting$ the number of conflicting agents (i.e., agents~$a$ with~$\enemies_a \neq \emptyset$);
    in the rows labeled using~$|\enemies_a|$, the stated bound on %
    $|\enemies_a|$
    applies to all agents ${a\in\agents}$.
    Entries of the form ``X\,/\,Y'' display the complexity for symmetric instances (left of the slash) and for asymmetric instances (right); a single entry applies to both settings.
    \NP-h (\W-h) means para-\NPhness (\Whness) for the given parameter value.}
    \label{tab:results}
    \centering\small
    \renewcommand{\arraystretch}{1.2}
    \newcommand{\poly}{\textsf{P}}
    \renewcommand{\NPc}{\NP-h}
    \renewcommand{\Wh}[1][1]{\W[#1]-h}
    \begin{tabularx}{\linewidth}{>{\centering}p{1.3cm}Y@{}Y@{}Y@{}Y@{}}
            \toprule
            & \IRkfo
            & \IRBSCfo
            & \IRFSCfo
            & \IRSCCfo
            \\\midrule
        $\numCoals = 1$ 
            & \poly\,{\tiny[O\ref{thm:FO:SCC:grandCoal:poly}]}
            & \poly\,{\tiny[O\ref{thm:FO:SCC:grandCoal:poly}]}
            & \poly\,{\tiny[O\ref{thm:FO:SCC:grandCoal:poly}]}
            & \poly\,{\tiny[O\ref{thm:FO:SCC:grandCoal:poly}]}
            \\
        $\numCoals \geq 2$
            & \NPc\,{\tiny[T\ref{thm:FO:k:twoCoals:NPh}]} 
            & \NPc\,{\tiny[T\ref{thm:FO:BSC:twoCoals:NPh}]}
            & \NPc\,{\tiny[T\ref{thm:FO:BSC:twoCoals:NPh}]}
            & \NPc\,{\tiny[T\ref{thm:FO:BSC:twoCoals:NPh}]} 
            \\\midrule
        $\ub\leq2$
            & $\nicefrac{\text{N}}{\text{A}}$ 
            & \poly\,{\tiny[T\ref{thm:FO:SCC:ubAtMostTwo:poly}]}
            & \poly\,{\tiny[T\ref{thm:FO:SCC:ubAtMostTwo:poly}]}
            & \poly\,{\tiny[T\ref{thm:FO:SCC:ubAtMostTwo:poly}]}
            \\
        $\ub\geq3$
            & $\nicefrac{\text{N}}{\text{A}}$ 
            & \NPc\,{\tiny[T\ref{thm:FO:BSC:ubAtMostThree:NPc}]}
            & \NPc\,{\tiny[T\ref{thm:FO:BSC:ubAtMostThree:NPc}]}
            & \NPc\,{\tiny[T\ref{thm:FO:BSC:ubAtMostThree:NPc}]}
            \\\midrule
        $r = 0$
            & \poly\,{\tiny[T\ref{thm:FO:k:symmetric:atLeastOneFriend:poly}]}/\NPc\,{\tiny[T\ref{thm:FO:k:asymmetric:atLeastOneFriend:NPh}]} 
            & \NPc\,{\tiny[T\ref{thm:FO:BSC:ubAtMostThree:NPc}]}
            & \NPc\,{\tiny[T\ref{thm:FO:BSC:ubAtMostThree:NPc}]}
            & \NPc\,{\tiny[T\ref{thm:FO:BSC:ubAtMostThree:NPc}]}
            \\
        $r = 1 $
            & ?/\NPc\,{\tiny[T\ref{thm:FO:k:asymmetric:atLeastOneFriend:NPh}]}
            & \NPc\,{\tiny[T\ref{thm:FO:BSC:ubAtMostThree:NPc}]}
            & \NPc\,{\tiny[T\ref{thm:FO:BSC:ubAtMostThree:NPc}]}
            & \NPc\,{\tiny[T\ref{thm:FO:BSC:ubAtMostThree:NPc}]}
            \\
        $r \geq 2$
            & \NPc\,{\tiny[T\ref{thm:FO:k:twoCoals:NPh}]} 
            & \NPc\,{\tiny[T\ref{thm:FO:BSC:ubAtMostThree:NPc}]}
            & \NPc\,{\tiny[T\ref{thm:FO:BSC:ubAtMostThree:NPc}]}
            & \NPc\,{\tiny[T\ref{thm:FO:BSC:ubAtMostThree:NPc}]}
            \\\midrule
        $\vc$
            & \FPT\,{\tiny[T\ref{thm:FO:k:treewidth:FPT}]}
            & \FPT\,{\tiny[T\ref{thm:FO:FSC:vc:FPT}]}
            & \FPT\,{\tiny[T\ref{thm:FO:FSC:vc:FPT}]} 
            & ? 
            \\
        $\tw+k$
            & \FPT\,{\tiny[T\ref{thm:FO:k:treewidth:FPT}]}
            & \Wh, \XP\,{\tiny[T\ref{thm:FO:SCC:treewidth:XP},T\ref{thm:FO:BSC:treedepth:Wh:collapsed}]}
            & \Wh, \XP\,{\tiny[T\ref{thm:FO:SCC:treewidth:XP},T\ref{thm:FO:BSC:treedepth:Wh:collapsed}]}
            & \Wh, \XP\,{\tiny[T\ref{thm:FO:SCC:treewidth:XP},T\ref{thm:FO:BSC:treedepth:Wh:collapsed}]}
            \\
        $\tw$
            & \FPT\,{\tiny[T\ref{thm:FO:k:treewidth:FPT}]}
            & \Wh\,{\tiny[T\ref{thm:FO:BSC:treedepth:Wh:collapsed}]}
            & \Wh \,{\tiny[T\ref{thm:FO:BSC:treedepth:Wh:collapsed}]}& \Wh \,{\tiny[T\ref{thm:FO:BSC:treedepth:Wh:collapsed}]}
            \\
        $\operatorname{cw}$
            & \Wh\,{\tiny[T\ref{thm:FO:k:cliquewidth:Wh:collapsed}]}
            & \NPc\,{\tiny[T\ref{thm:FO:k:cliquewidth:Wh:collapsed}]}
            & \NPc\,{\tiny[T\ref{thm:FO:k:cliquewidth:Wh:collapsed}]}
            & \NPc\,{\tiny[T\ref{thm:FO:k:cliquewidth:Wh:collapsed}]}
            \\\midrule
        $\vc^\enemies=1$
            & \poly\,{\tiny[T\ref{thm:FO:k:enmityVcAtMostOne:poly}]}
            & \poly\,{\tiny[T\ref{thm:SSC:k:symmetric:enmityVcAtMostOne:poly}]}/\NPc\,{\tiny[T\ref{thm:BSC:k:enmityVcAtMostOne:NPh}]}
            & \poly\,{\tiny[T\ref{thm:SSC:k:symmetric:enmityVcAtMostOne:poly}]}/\NPc\,{\tiny[T\ref{thm:BSC:k:enmityVcAtMostOne:NPh}]}
            & \poly\,{\tiny[T\ref{thm:SSC:k:symmetric:enmityVcAtMostOne:poly}]}/\NPc\,{\tiny[T\ref{thm:BSC:k:enmityVcAtMostOne:NPh}]}
            \\
        $\vc^\enemies=2$
            & \poly\,{\tiny[T\ref{thm:FO:k:symmetric:enmityVcAtMostThree:poly}]}/\NPc\,{\tiny[T\ref{thm:FO:k:enmityVcAtLeastTwo:NPh:collapsed}]}
            & ?/\NPc\,{\tiny[T\ref{thm:FO:k:enmityVcAtLeastTwo:NPh:collapsed}]}
            & ?/\NPc\,{\tiny[T\ref{thm:FO:k:enmityVcAtLeastTwo:NPh:collapsed}]}
            & ?/\NPc\,{\tiny[T\ref{thm:FO:k:enmityVcAtLeastTwo:NPh:collapsed}]}
            \\
        $\vc^\enemies=3$
            & \poly\,{\tiny[T\ref{thm:FO:k:symmetric:enmityVcAtMostThree:poly}]}/\NPc\,{\tiny[T\ref{thm:FO:k:enmityVcAtLeastTwo:NPh:collapsed}]}
            & \NPc\,{\tiny[T\ref{thm:FO:BSC:symmetric:enmityVcAtLeastThree:NPh:collapsed}]}
            & \NPc\,{\tiny[T\ref{thm:FO:BSC:symmetric:enmityVcAtLeastThree:NPh:collapsed}]}
            & \NPc\,{\tiny[T\ref{thm:FO:BSC:symmetric:enmityVcAtLeastThree:NPh:collapsed}]}
            \\
        $\vc^\enemies\geq 4$
            & \NPc\,{\tiny[T\ref{thm:FO:k:enmityVcAtLeastFour:NPh:collapsed}]}
            & \NPc\,{\tiny[T\ref{thm:FO:k:enmityVcAtLeastFour:NPh:collapsed}]}
            & \NPc\,{\tiny[T\ref{thm:FO:k:enmityVcAtLeastFour:NPh:collapsed}]}
            & \NPc\,{\tiny[T\ref{thm:FO:k:enmityVcAtLeastFour:NPh:collapsed}]}
            \\\midrule
        $\numConflicting$
            & \FPT\,{\tiny[T\ref{thm:FO:k:numConflicting:FPT}]}
            & \FPT\,{\tiny[T\ref{thm:FO:k:numConflicting:FPT}]}
            & \FPT\,{\tiny[T\ref{thm:FO:k:numConflicting:FPT}]}
            & ?
            \\\midrule
        $|\enemies_a| = 0$
            & \poly\,{\tiny[O\ref{thm:FO:SCC:noEnemies:poly}]}
            & \poly\,{\tiny[O\ref{thm:FO:SCC:noEnemies:poly}]}
            & \poly\,{\tiny[O\ref{thm:FO:SCC:noEnemies:poly}]}
            & \poly\,{\tiny[O\ref{thm:FO:SCC:noEnemies:poly}]}
            \\
        $|\enemies_a| = 1$
            & \poly\,{\tiny[T\ref{thm:FO:k:oneEnemy:poly}]}/%
            \NPc\,{\tiny[T\ref{thm:FO:k:enmityVcAtLeastTwo:NPh:collapsed}]}
            & \poly\,{\tiny[T\ref{thm:FO:BSC:oneEnemy:poly}]}/%
            \NPc\,{\tiny[T\ref{thm:BSC:k:enmityVcAtMostOne:NPh}]}
            & ?/\NPc\,{\tiny[T\ref{thm:BSC:k:enmityVcAtMostOne:NPh}]}  
            & ?/\NPc\,{\tiny[T\ref{thm:BSC:k:enmityVcAtMostOne:NPh}]} \\
        $|\enemies_a| = 2$
            & ?/\NPc\,{\tiny[T\ref{thm:FO:k:enmityVcAtLeastTwo:NPh:collapsed}]} & ?/\NPc\,{\tiny[T\ref{thm:BSC:k:enmityVcAtMostOne:NPh}]} & ?/\NPc\,{\tiny[T\ref{thm:BSC:k:enmityVcAtMostOne:NPh}]} & ?/\NPc\,{\tiny[T\ref{thm:BSC:k:enmityVcAtMostOne:NPh}]} \\
        $|\enemies_a| = 3$
            & ?/\NPc\,{\tiny[T\ref{thm:FO:k:enmityVcAtLeastTwo:NPh:collapsed}]} 
            & ?/\NPc\,{\tiny[T\ref{thm:BSC:k:enmityVcAtMostOne:NPh}]} 
            & \NPc\,{\tiny[T\ref{thm:FO:FSC:threeEnemies:NPh}]} 
            & \NPc\,{\tiny[T\ref{thm:FO:FSC:threeEnemies:NPh}]} 
            \\
        $|\enemies_a| \geq 4$
            & \NPc\,{\tiny[T\ref{thm:FO:k:fourEnemies:NPh}]} 
            & \NPc\,{\tiny[T\ref{thm:FO:k:fourEnemies:NPh}]} 
            & \NPc\,{\tiny[T\ref{thm:FO:k:fourEnemies:NPh}]} 
            & \NPc\,{\tiny[T\ref{thm:FO:k:fourEnemies:NPh}]}
            \\\bottomrule
    \end{tabularx}%
\end{table}

\subsection{Our Contribution}
We provide a detailed study of individual rationality under the two structural constraints described above, considering four variants that differ in how the coalition sizes are restricted. In increasing order of generality, these are: the \emph{$\numCoals$-coalition} variant, which only fixes the number of non-empty coalitions; the \emph{balanced-size} (BSC) variant, where any two coalition sizes may differ by at most one; the \emph{fixed-size} (FSC) variant, where each coalition must have a prescribed size; and the \emph{size-constrained} (SCC) variant, where each coalition is assigned its own lower and upper bounds. For each of these, we chart the complexity of deciding whether an individually rational outcome exists, as summarized in \Cref{tab:results}. We approach the problem from several dimensions: the number of desired coalitions and the permitted coalition sizes; whether relationships are symmetric or asymmetric;   the maximum number of enemies per agent; the number of conflicting agents (those having enemies); the number of recluses (agents having enemies but no friends); and structural parameters of the underlying relationship and enmity graphs, namely their vertex cover number, treewidth, and clique-width.

Our first finding is that the two preference orientations could hardly behave more differently. Under enemy-oriented preferences, friendships provably play no role: all four variants collapse to size-constrained versions of the classic \probName{$k$-Coloring} problem (\Cref{sec:enemy-oriented}), and their complexity is thus inherited from the graph-coloring literature. Under friend-oriented preferences, no such collapse occurs---the complexity is governed by the enmity structure in a far subtler way, and charting this landscape constitutes the technical core of the paper.

Since friend-oriented (and enemy-oriented) preferences form a special case of additively separable hedonic games, our results directly extend those of Fioravantes \emph{et al.}~\cite{FioravantesGMS2026a} to this prominent preference domain. Our techniques, however, are quite different. Whereas Fioravantes \emph{et al.} build their algorithms largely on $N$-fold integer formulations, we draw on a broad range of tools from the parameterized complexity toolbox---including color coding, bounded search trees, and dynamic programming over nice tree decompositions. On the negative side, our hardness reductions start from a variety of source problems and, owing to the nuanced picture we paint, require careful, technically involved constructions.

\subsection{Related Work}
\emph{Friends-and-enemies games} were introduced by Dimitrov \emph{et al.}~\cite{DimitrovBHS2006}. Ohta \emph{et al.}~\cite{OhtaBISY2017} later extended the model with neutral agents. The existence and verification of stable outcomes in this domain have since been studied for a range of stability notions~\cite{LangRRSS2015,RotheSS2018,ChenCRS2023}. Importantly, Chen \emph{et al.}~\cite{ChenCRS2023} studied parameterized complexity of deciding the existence of stable solutions for the friends--enemies--neutrals model. None of these works, however, considers individual rationality and/or structural constraints on the outcome.

Friends-and-enemies games form a subclass of \emph{additively separable hedonic games} (ASHGs)~\cite{BogomolnaiaJ2002} which, together with \emph{fractional hedonic games}, (FHGs)~\cite{AzizBBHOP2019} are among the most studied classes of hedonic games; see, e.g., \cite{PetersE2015,BrandtBT2024,FioravantesGM2025}. As IR is trivially satisfiable in any such game without further restrictions, these works focus on more demanding notions of stability. An exception is the work of Fioravantes \emph{et al.}~\cite{FioravantesGMS2026a}, who studied the same model as we do, albeit for the more general class of ASHGs. Their results, however, do not pertain to our setting: their hardness results do not transfer to a restricted preference domain, and since expressing friend- or enemy-oriented preferences by additive valuations requires weights linear in the number of agents, their algorithmic results---which achieve fixed-parameter tractability only when the maximum weight is taken as an additional parameter---yield at best \XP{} algorithms in our setting. A similar model was also studied by Fioravantes \emph{et al.}~\cite{FioravantesGMS2026b}; however, for a different preference model incomparable with ours.

More broadly, across the many subclasses of hedonic games, individual rationality is rarely studied as a solution concept: either no coalition is worse than the singleton, so IR is trivially achievable even under our constraints, or it was simply not explored. The one exception we are aware of is the work of Deligkas \emph{et al.}~\cite{DeligkasEKS2024b}, who study IR in \emph{topological distance games}~\cite{BullingerS2023}; that model is incomparable to ours, and their (mostly negative) results do not carry over.

Restrictions on the admissible coalition structures other than ours have also been considered. Beyond the $\numCoals$-coalition and fixed-size models discussed above, Wright and Vorobeychik~\cite{WrightV2015} introduced a bounded maximum coalition size, studied further in several follow-ups~\cite{LevingerAH2023,GanianHKSS2023,FioravantesGM2025}. Even under this restriction, however, IR remains trivial, as one may always form arbitrarily many singletons.

\section{Preliminaries}
\label{sec:preliminaries}
\appendixsection{sec:preliminaries}

We assume the reader is familiar with the basics of graph theory~\cite{diestel-book} and both classical and parameterized complexity~\cite{cyg-fom-kow-lok-mar-pil-pil-sau:b:parameterized-algorithms}\iflong{}; in \Cref{sec:param_complexity}, we give a brief overview of the parameterized-complexity notions we use in this work.\else{}.\fi{}

We consider a set $\agents$ of $\numAgents$ \emph{agents}. Each agent $a\in\agents$ is associated with a set of \emph{friends} $\friends_a\subseteq \agents\setminus\{a\}$ and a set of \emph{enemies} $\enemies_a\subseteq \agents\setminus\{a\}$ so that $\friends_a \cap \enemies_a = \emptyset$. Note that $(\friends_a,\enemies_a)$ is not necessarily a partition of $\agents\setminus\{a\}$:
agents in $\agents\setminus ( \friends_a \cup \enemies_a \cup \{a\} )$ are \emph{neutrals} for~$a$; we denote their set 
as~$\neutrals_a$. We can encode the friends--enemies--neutral relations using a directed edge-colored \emph{relationship graph} $G=(\agents,E,\cFn)$, where for each $a$ and every $b\in \friends_a \cup \enemies_a$, $E$ contains an arc~$(a,b)$. This arc is colored \emph{blue} if $b$ is a friend of $a$ and \emph{red} if $b$ is~$a$'s enemy.
We use $G^\friends$ to denote the \emph{friendship graph}, i.e., the graph $G$ reduced to the blue arcs, and we use $G^\enemies$ to denote the \emph{enmity graph} constructed from~$G$ by leaving only the red arcs.
Agents' relations are \emph{symmetric} if $b \in \friends_a \iff a \in\friends_b$ and $b \in \enemies_a \iff a \in \enemies_b$ hold for every 
$a,b\in\agents$. In this case, the graph $G$ is undirected.  

Any subset $C \subseteq A$ of agents $\agents$ is called a \emph{coalition}. We say that $C\subseteq \agents$ is a \emph{singleton}  if $|C| = 1$, and  $C$ is the \emph{grand coalition} if $C=\agents$. By $\mathcal{C}_a$ we denote the set of all coalitions that contain agent $a\in\agents$. The preferences of an agent $a\in\agents$ is a weak ordering $\succeq_a$ over $\mathcal{C}_a$. For two coalitions $C,D \in \mathcal{C}_a$, we use $C \succ_a D$ to represent that agent $a$ strictly prefers coalition $C$ over $D$, i.e., $C \succeq_a D$ and $D \not\succeq_a C$. Similarly, we use $C \sim D$ if agent~$a$ is indifferent between~$C$ and~$D$, that is, $C \succeq_a D$ and $D \succeq_a C$. A list $\mathcal{P} = (\succeq_a)_{a\in\agents}$ is called the \emph{preference profile}. Throughout the paper, we consider the following two restrictions of agents' preferences based on the sets of friends and enemies:

\begin{itemize}
    \item A preference profile $\mathcal{P}$ is \emph{friend-oriented} if for every $a\in\agents$ and each pair of~$C,D\in\mathcal{C}_a$ it holds that
    $C \succeq_a D$ if and only if (1) $|C \cap \friends_a| > |D \cap \friends_a|$ or (2)~$|C \cap \friends_a| = |D \cap \friends_a|$ and $|C \cap \enemies_a| \leq  |D \cap \enemies_a|$.
    \item A preference profile $\mathcal{P}$ is \emph{enemy-oriented} if for every $a\in\agents$ and each pair of~$C,D\in\mathcal{C}_a$ it holds that
    $C \succeq_a D$ if and only if (1) $|C \cap \enemies_a| < |D \cap \enemies_a|$ or (2)~$|C \cap \enemies_a| = |D \cap \enemies_a|$ and $|C \cap \friends_a| \geq  |D \cap \friends_a|$.
\end{itemize}

A \emph{coalition structure}~$\pttn = (\pttn_1,\ldots,\pttn_\numCoals)$ is a partition of agents into coalitions, that is, $\bigcup_{i=1}^\numCoals \pttn_i = \agents$ and $\pttn_i \cap \pttn_{j} = \emptyset$ for each pair of distinct $i,j\in[\numCoals]$. For an agent $a\in\agents$ and a coalition structure $\pttn$, we use $\pttn(a)$ to denote the coalition that agent $a$ is part of. We extend the preferences from coalitions to coalition structures by setting $\pttn \succeq_a \pttn'$ whenever $\pttn(a) \succeq_a \pttn'(a)$.

The goal of a \emph{hedonic game} $\Gamma = (\agents,\mathcal{P})$ is to find a \emph{desirable} coalition structure $\pttn$, where \emph{desirable} is some well-defined requirement on the solution coalition structure $\pttn$. 
In this paper, the desirable coalition structures are those that meet certain \emph{stability} criteria. The notion of stability that we adopt in this work is that of \emph{individual rationality} (IR), which requires that no agent strictly prefers to be in a singleton coalition over staying in their current coalition.

\begin{definition}
    Let $\Gamma$ be a hedonic game, and $\pttn$ a coalition structure for~$\Gamma$. We say that $\pttn$ is \emph{individually rational} (IR) if $\pttn(a) \succeq_a \{a\}$ for every agent~$a\in\agents$.
\end{definition}

For our two preference restrictions, individual rationality admits a simple combinatorial characterization, which we use throughout the paper in place of the preference relations themselves.

\begin{observation}\label{obs:IR:characterization}
    Let $\Gamma=(A,\mathcal{P})$ be a hedonic game and $\pttn$ a coalition structure.%
    \begin{enumerate}[label=(\alph*)]
        \item If $\mathcal{P}$ is friend-oriented, then $\pttn$ is IR if and only if for each agent $a \in A$ we have  that $\pttn(a)$  contains a friend of~$a$ or no enemy of~$a$. 
        \item If $\mathcal{P}$ is  enemy-oriented, then $\pttn$ is IR if and only if for each agent $a \in A$ we have  that  $\pttn(a)$ contains no enemy of~$a$.
    \end{enumerate}
\end{observation}

Nevertheless, it is trivial to find an individually rational coalition structure for any hedonic game---simply place each agent into a coalition by themselves. Therefore, we consider the following four natural \emph{constraints}:%

\begin{enumerate}
    \item A \emph{hedonic game with size-constrained coalition} (HG-SCC) is a tuple $(\agents,\mathcal{P},\numCoals,\lb,\ub)$ where $\numCoals \in [\numAgents]$ is the desired number of coalitions, and $\lb$ and $\ub$ are functions from $[\numCoals]$ to~$\{0,1,\dots,\numAgents\}$ setting a lower and upper bound on the size of each coalition. The solution structure~$\pttn$ of such a game must satisfy
    $\pttn = (\pttn_1,\ldots,\pttn_\numCoals)$   where $ \lb(j) \leq |\pttn_j| \leq \ub(j)$ for each $j \in [\numCoals]$.

    \item A \emph{hedonic game with $\numCoals$ coalitions} ($\numCoals$-HG) is HG-SCC where $\lb(j) = 1$ and $\ub(j) = \numAgents$ for every $j \in [\numCoals]$.

    \item A \emph{hedonic game with fixed-size coalitions} (HG-FSC) is HG-SCC with $\lb(j) = \ub(j)$ for every $j\in[\numCoals]$. %

    \item A \emph{hedonic game with balanced-size coalitions} (HG-BSC) is HG-FSC such that $\ub(j) \in \{ \lfloor \nicefrac{\numAgents}{\numCoals} \rfloor, \lceil \nicefrac{\numAgents}{\numCoals} \rceil \}$ for every $j\in[\numCoals]$.
\end{enumerate}
\toappendix{
\begin{figure}
    \centering
    \begin{tikzpicture}
        \node[draw,rounded corners] (SCC) at (0,0) {HG-SCC};
        \node[draw,rounded corners] (k) at (-1,-2) {$\numCoals$-HG};
        \node[draw,rounded corners] (FSC) at (1,-1) {HG-FSC};
        \node[draw,rounded corners] (BSC) at (1,-2) {HG-BSC};

        \draw[->] (k) -- (SCC);
        \draw[->] (BSC) -- (FSC);
        \draw[->] (FSC) -- (SCC);
    \end{tikzpicture}
    \caption{A relation between constraints we consider in our work. An arrow from constraint $A$ to constraint $B$ represents that any solution for $A$ is also a solution for $B$. Consequently, any hardness result for $A$ implies the same hardness result for $B$ and any algorithmic result for $B$ implies the same algorithmic result for $A$.}
    \label{fig:constraintsRelationships}
\end{figure}

In this section, we provide \Cref{fig:constraintsRelationships}, illustrating the relationship between the problem variants we study, and then we give some background on the notions we use from the parameterized complexity framework. 

}
The relationships between the constraints are as follows: HG-BSC is a special case of HG-FSC, which in turn is a special case of HG-SCC. $\numCoals$-HG is a special case of HG-SCC and is incomparable with any other restriction. Naturally, a hardness result for a special case implies hardness for a more general model, and algorithmic upper bounds are implied in the opposite direction.

Given an HG-SCC~$\Gamma$ with friend-oriented preferences, the problem of deciding whether $\Gamma$ admits an IR coalition structure is denoted \IRSCCfo; the problems \IRFSCfo, \IRBSCfo, and \IRkfo, as well as their enemy-oriented variants, are defined analogously. For fixed-size (or balanced-size) coalitions, we will assume that the given sizes sum up to~$\numAgents$. 
Since we can check efficiently whether a coalition structure is IR, all studied problems are in~\NP.

\toappendix{
\subsection{Parameterized Complexity}\label{sec:param_complexity}

We apply the framework of parameterized complexity where each input instance~$I$ of a parameterized problem~$Q$ comes with an integer parameter~$q$, and the running time of an algorithm solving~$Q$ is not only measured as a function of the input length~$|I|$ but rather as a function of both~$|I|$ and~$q$. When dealing with a problem that is intractable in the classical sense, the usual aim is to restrict the combinatorial explosion to the parameter and, accordingly, design an algorithm with running time $f(q) \cdot |I|^{\Oh{1}}$ for some computable function~$f$; such an algorithm is called \emph{fixed-parameter tractable}, and the corresponding complexity class is \FPT. By contrast, an algorithm whose running time is polynomial for each fixed constant value of the parameter is said to be an \XP algorithm; note that here the exponent of the polynomial is not necessarily a constant but may depend on~$q$.
A problem is \emph{\pNPh} if it is \NPh for some constant value of the parameter; clearly, we cannot expect such a problem to admit an \XP algorithm. 
As an analog of \NPhness, parameterized complexity theory uses the notion of \Whness[1], based on so-called \emph{parameterized reductions} and the complexity class \W[1]. 
Hence, showing the \Whness[1] of a parameterized problem~$Q$ is strong evidence that $Q \notin \FPT$.  
We refer the reader to \cite{cyg-fom-kow-lok-mar-pil-pil-sau:b:parameterized-algorithms} for a more comprehensive introduction to the topic.%
}

\toappendix{
\paragraph{Structural parameters.}

When studying the structure of the relationship graph~$G$ of a hedonic game, we will focus on some of the most widely used structural parameters in graph theory. Arguably, the most prominent of such parameters is the treewidth of~$G$, denoted by~$\tw(G)$, which measures the tree-likeness of~$G$. A \emph{tree-decomposition} for~$G$ is a tree~$\T$ and a function~$\beta:V(\T) \rightarrow 2^{V(G)}$ that assigns a \emph{bag}~$\beta(x)$ to each \emph{node}~$x \in V(\T)$ such that 
\begin{itemize}
    \item for each $v \in V(G)$ there is a node~$x \in V(T)$ with $v \in  \beta(x)$, 
    \item for each $\{u,v\} \in E(G)$ there is a node~$x \in V(\T)$ with $\{u,v\} \subseteq \beta(x)$, and 
    \item for each $v \in V(G)$, the nodes in $\{x: v \in \beta(x)\}$ induce a subtree of~$\T$.
\end{itemize}
The \emph{width} of~$\T$ is $\max_{x \in V(\T)} |\beta(x)|-1$, and the \emph{treewidth} of~$G$ is the minimal width of a tree-decomposition for~$G$. 
A graph parameter that is stronger than treewidth is the \emph{vertex cover number} of~$G$, denoted by~$\vc(G)$, which is the size of a minimum vertex cover of~$G$, where a \emph{vertex cover} of~$G$ is a set~$U$ of vertices such that every edge has an endpoint in~$U$.
It is easy to observe that $\tw(G) \leq \vc(G)$: taking a minimum vertex cover~$C$ of~$G$ and creating one bag~$C \cup \{v\}$ for each vertex $v \notin C$ yields a (path) decomposition of~$G$ of width~$|C|$.
Although some of our results concern the structural parameter \emph{clique-width}, the precise definition will not be needed. However, note that $\operatorname{cw}(G) \leq 3\cdot 2^{\tw(G)-1 }$.
}

\section{Enemy-Oriented Preferences}
\label{sec:enemy-oriented}
\appendixsection{sec:enemy-oriented}
\toappendix{This section contains all proofs omitted from \Cref{sec:enemy-oriented}.}

Our first result settles the enemy-oriented case completely, by showing that it collapses to a classic coloring problem. Indeed, if preferences are enemy-oriented, then friendships play no role in whether a given coalition structure is IR: being alone is \emph{not} preferred by an agent~$a$ to being in a coalition~$C$ if and only if $a$ has no enemies in~$C$. Therefore, individual rationality of a coalition structure~$\pi$ is equivalent to each coalition forming an independent set in the enmity graph~$G^\enemies$; 
notice that enmities can be assumed to be symmetric. That is, our problems are equivalent to the size-constrained variants of the classic \probName{$k$-Coloring} problem. 
Since \probName{Equitable $3$-Coloring} is \NPc already for $4$-regular graphs, this immediately implies the following:%
    
\begin{restatable}[\linkproof{obs:EO:kthree:NPh}]{observation}{obsEOkthreeNPh}
\label{obs:EO:kthree:NPh}
    \IRkeo and \IRBSCeo are \NPc
    even if ${k=3}$ and $|\enemies_a|=4$ for every $a \in \agents$. 
\end{restatable}
\prooftoappendix{obs:EO:kthree:NPh}{\obsEOkthreeNPh*}{
    Notice that the reduction presented in \Cref{thm:FO:k:fourEnemies:NPh} to prove the \NPcness of \IRkfo and \IRBSCfo for $\numCoals=3$ results in an instance without friendships; hence, friend- and enemy-oriented preferences coincide. This implies our result.
}

By contrast, it is not hard to see that deciding whether a graph~$H$ can be $2$-colored so that the color classes respect some size constraint can be done in polynomial time, using a simple dynamic programming approach based on the observation that each connected component of~$H$ has a unique $2$-coloring up to the reversal of the colors.
This yields the following result.

\begin{restatable}[\linkproof{thm:EO:SCC:kTwo:poly}]{theorem}{thmEOSCCkTwopoly}
\label{thm:EO:SCC:kTwo:poly}
    \IRSCCeo for $k=2$ is polynomial-time solvable. 
\end{restatable}

\prooftoappendix{thm:EO:SCC:kTwo:poly}{\thmEOSCCkTwopoly*}{
    Notice that the problem can be reduced to the following: given a graph~$H$ and integers~$s_1$ and~$s_2$, decide if $G$ admits a proper $2$-coloring where the color classes have size~$s_1$ and~$s_2$ (with $s_1+s_2=|V(H)|$).  
    
    Let $H$ consist of connected components $H_1,\dots,H_\ell$, then for each $i=1,\dots,\ell$ we store  
    the set~$S_i$ of integers~$s$ for which $\bigcup_{j \in [i]} H_1$ can be $2$-colored with one color class having size exactly~$s$. Note that for each $H_i$, there is only one way to $2$-color it up to switching the two colors. We  compute such a coloring in linear time; let $t_i$ and~$t'_i=|V(H_i)|-t_i$ denote the sizes of the two color classes we obtain.
    Then for $i=1$ we know $S_1=\{t_1,t'_1\}$.  
    For $i>1$, we can compute $S_i$ using $S_i=\{s+t \colon s \in S_{i-1}, t \in \{t_i,t'_i\}\}$. This way, we can decide if $H$ can be $2$-colored with color classes of sizes~$s_1$ and~$s_2$ by checking if $s_1 \in S_\ell$. The presented algorithm runs in linear time.
}

The polynomial equivalence with size-constrained coloring settles the enemy-oriented case: any complexity result for these coloring problems, classical or parameterized, translates directly to our setting, and vice versa. No such collapse is possible under friend-oriented preferences, where friendships can compensate for the presence of enemies; the resulting---much richer---complexity landscape is the subject of the rest of the paper.

\section{Friend-Oriented Preferences}
\label{sec:FO}
\appendixsection{sec:FO}
\toappendix{This section contains all proofs omitted from \Cref{sec:FO}.}

Directing our attention to friend-oriented preferences, we investigate how different natural parameters influence the computational complexity of our problems.

\subsection{Few or Small Coalitions}
We show that our problems are already hard if we aim only for two coalitions, yielding a strict dichotomy in terms of the number of desired coalitions.

\begin{restatable}[\linkproof{thm:FO:k:twoCoals:NPh}]{theorem}{thmFOktwoCoalsNPh}
\label{thm:FO:k:twoCoals:NPh}\label{thm:FO:BSC:twoCoals:NPh}\label{thm:FO:SCC:grandCoal:poly}
    If $\numCoals = 1$, \IRSCCfo  can be decided in linear time.
    \IRkfo and \IRBSCfo are \NPc for $\numCoals=2$, even if preferences are symmetric, each agent has at most $8$ friends, and there are only two agents with no friends. 
\end{restatable}

\begin{proofsketch}
    To show \NPhness of \IRkfo for $k = 2$, we give a reduction from \NPh \probName{Monotone $4$-Regular NAE 3-SAT}~\cite{DarmannD2020}. The input of this problem is 3-CNF formula $\varphi=C_1 \wedge \dots \wedge C_m$ over a set $X=\{x_1,\dots,x_n\}$ of variables %
    such that $\varphi$ contains no negative literals and 
    each variable appears in exactly four clauses. The question is whether~$\varphi$ admits a \emph{valid} truth assignment, i.e., one in which each clause is satisfied by some but not all of its variables. %

    \begin{figure}[bt!]
        \centering
        \scalebox{0.8}{
        \begin{tikzpicture}[
                agent/.style={circle, draw=black, minimum size=20pt, inner sep=1pt, font=\small,fill=white},
                friend/.style={draw=black,cbBlue, thick, dashed},
                enemy/.style={draw=black,cbRed,thick},
                every node/.style={font=\small},
                var/.style={minimum size=15pt}
            ]
            \pgfmathsetmacro\yPos{1.5}
            
            \node[agent,thick] (rp) at (0,\yPos)  {$r^+$};
            \node[agent,thick] (rn) at (0,-\yPos) {$r^-$};
            
            \node[agent] (c1n) at (3,\yPos) {$c_1^-$};
            \node[agent] (c2n) at (5,\yPos) {$c_2^-$};
            \node[agent] (c3n) at (7,\yPos) {$c_3^-$};
            
            \node[agent] (c1p) at (3,-\yPos) {$c_1^+$};
            \node[agent] (c2p) at (5,-\yPos) {$c_2^+$};
            \node[agent] (c3p) at (7,-\yPos) {$c_3^+$};
            
            \node[agent, var] (a1) at (2, 0)  {$a_1$};
            \node[agent, var] (a2) at (3.5, 0) {$a_2$};
            \node[agent, var] (a3) at (5, 0)  {$a_3$};
            \node[agent, var] (a4) at (6.5, 0) {$a_4$};
            \node[agent, var] (a5) at (8, 0)  {$a_5$};

            \draw[friend] (c1p) -- (a1);
            \draw[friend] (c1p) -- (a2);
            \draw[friend] (c1p) -- (a3);
            \draw[friend] (c1n) -- (a1);
            \draw[friend] (c1n) -- (a2);
            \draw[friend] (c1n) -- (a3);
            
            \draw[friend] (c2p) -- (a1);
            \draw[friend] (c2p) -- (a3);
            \draw[friend] (c2p) -- (a4);
            \draw[friend] (c2n) -- (a1);
            \draw[friend] (c2n) -- (a3);
            \draw[friend] (c2n) -- (a4);
            
            \draw[friend] (c3p) -- (a2);
            \draw[friend] (c3p) -- (a4);
            \draw[friend] (c3p) -- (a5);
            \draw[friend] (c3n) -- (a2);
            \draw[friend] (c3n) -- (a4);
            \draw[friend] (c3n) -- (a5);

            \draw[enemy] (rp) -- (rn);
            
            \draw[enemy,bend left=22] (rp) to (c1n);
            \draw[enemy,bend left=22] (rp) to (c2n);
            \draw[enemy,bend left=22] (rp) to (c3n);
            
            \draw[enemy,bend right=22] (rn) to (c1p);
            \draw[enemy,bend right=22] (rn) to (c2p);
            \draw[enemy,bend right=22] (rn) to (c3p);
            
            \draw[enemy] (c1n) to (c2n);
            \draw[enemy] (c2n) to (c3n);
            
            \draw[enemy] (c1p) to (c2p);
            \draw[enemy] (c2p) to (c3p);

            \begin{pgfonlayer}{background}
                \filldraw[thick,dotted,cbOrange,fill=cbOrange!7,rounded corners=10pt]
                    (1, -2.5) -- (1,-0.5) -- (-1.25,-0.5) -- (-1.25,2.75) --
                    (2.5,2.75) -- (2.5,0.75) -- (4.5,0.75) -- (4.5,-0.5) -- 
                    (5.5,-0.5) -- (5.5,0.75) -- (7.5,0.75) -- (7.5,-0.5) --
                    (8.5,-0.5) -- (8.5,3) -- (-1.5,3) -- (-1.5,-2.5) -- cycle;
                \node[cbOrange] at (8,2.5) {\large$\pttn_2$};

                \filldraw[thick,dotted,cbGreen,fill=cbGreen!7,rounded corners=10pt]
                    (-1,2.5) -- (2,2.5) -- (2,0.5) -- (4,0.5) -- (4,-0.75) --
                    (6,-0.75) -- (6,0.5) -- (7,0.5) -- (7,-0.75) -- (8.5,-0.75) --
                    (8.5,-2.5) -- (1.5,-2.5) -- (1.5,0) -- (-1,0) -- cycle;
                \node[cbGreen] at (8,-2) {\large$\pttn_1$};
            \end{pgfonlayer}
        \end{tikzpicture}
        }
        \caption{An illustration of the hardness construction used to prove \Cref{thm:FO:k:twoCoals:NPh} for a formula $\varphi = (x_1 \lor x_2 \lor x_3) \land (x_1 \lor x_3 \lor x_4) \land (x_2 \lor x_4 \lor x_5)$. Red (solid) edges represent that two agents are enemies, while blue (dashed) edges represent friendship. Using colored backgrounds, we highlight one possible IR coalition structure.}
        \label{fig:thm:FO:k:twoCoals:NPh:construction}
    \end{figure}

    Given $\varphi$, we construct an equivalent instance $\mathcal{J}$ of FO-$\numCoals$-HG as follows (see \Cref{fig:thm:FO:k:twoCoals:NPh:construction} for an illustration of the construction). We have one \emph{variable agent} $a_i$ for every variable $x_i\in X$, and two \emph{clause agents} $c_j^+$ and $c_j^-$ for every clause~$C_j$. In addition, we create two \emph{guard agents} $r^+$ and $r^-$, whom we let be enemies. %
    We also let $r^+$ be enemies with all clause agents~$c_j^-$, and similarly, $r^-$ be enemies with all clause agents $c_j^+$. The guard agents have no friends. Furthermore, for each $j\in[m-1]$ the pairs $(c_j^-,c_{j+1}^-)$ and $(c_j^+,c_{j+1}^+)$ of clause agents are also enemies. Each clause agent~$c_j^+$ or $c_j^-$ has three variable agents as their friends, namely the ones corresponding to literals present in the clause $C_j$. All remaining relations are neutral. To finalize the construction, we set $\numCoals = 2$.
    
    For correctness, assume first that $\alpha$ is a valid assignment for $\varphi$. 
    Then the coalition structure $\pttn$ defined by $\pttn_1 = \{ r^+ \} \cup \{ c_j^+ \colon j\in[m] \} \cup \{ a_i \colon \alpha( x_i ) = \texttt{true} \}$ and $\pttn_2 = \agents \setminus \pttn_2$ is IR for~$\mathcal{J}$, as can be checked in a straightforward way.

    In the opposite direction, let $\pttn=(\pttn_1,\pttn_2)$ be an IR coalition structure for~$\mathcal{J}$. We first observe that $r^+$ and $r^-$ cannot be members of the same coalition, since they do not have friends and are mutual enemies. By the same argument, no clause agent~$c^-_j$ can be in the same coalition with $r^+$, and no clause agent~$c^+_j$ can be in the same coalition as $r^-$. Hence, without loss of generality, we can assume that $\{r^+,c_1^+,\ldots,c_m^+\} \subseteq \pttn_1$ and $\{r^-,c_1^-,\ldots,c_m^-\}\subseteq \pttn_2$. We construct an assignment $\alpha$ by setting $\alpha(x_i) = \texttt{true}$ if and only if $a_i\in\pttn_1$ for every $i\in[n]$. For contradiction, assume that $\alpha$ is not valid for $\varphi$, i.e., there is a clause $C_j$ with either all or none of its literals being true. In the former case,  all friends of~$c_j^-$ are in~$\pttn_1$. However, $c_j^-$ has at least one enemy in $\pttn_2$, which contradicts the individual rationality of~$\pttn$. In the latter case, we apply the same argument on~$c_j^+$. Consequently, $\alpha$ is a valid truth assignment, finishing the proof for \IRkfo. 
     For \IRBSCfo, we simply add $n$ dummy agents $d_1,\ldots,d_n$ that are neutrals from the perspective of all agents. %
\end{proofsketch}
\prooftoappendix{thm:FO:k:twoCoals:NPh}{\thmFOktwoCoalsNPh*}{
    If $k=1$, the only acceptable coalition is the grand coalition, and we can check whether it is IR in polynomial time.
    To show \NPhness of \IRkfo for $k = 2$, we give a reduction from \NPh \probName{Monotone $4$-Regular NAE 3-SAT}~\cite{DarmannD2020}. The input of this problem is 3-CNF formula $\varphi=C_1 \wedge \dots \wedge C_m$ over a set $X=\{x_1,\dots,x_n\}$ of variables 
    such that $\varphi$ contains no negative literals and 
    each variable appears in exactly four clauses. The question is whether~$\varphi$ admits a \emph{valid} truth assignment, i.e., one in which each clause is satisfied by some but not all of its variables. 

    Given $\varphi$, we construct an equivalent instance $\mathcal{J}$ of FO-$\numCoals$-HG as follows (see \Cref{fig:thm:FO:k:twoCoals:NPh:construction} for an illustration of the construction). We have one \emph{variable agent} $a_i$ for every variable $x_i\in X$, and two \emph{clause agents} $c_j^+$ and $c_j^-$ for every clause~$C_j$. In addition, we create two \emph{guard agents} $r^+$ and $r^-$, whom we let be enemies.  
    We also let $r^+$ be enemies with all clause agents~$c_j^-$, and similarly, $r^-$ be enemies with all clause agents $c_j^+$. The guard agents have no friends. Furthermore, for each $j\in[m-1]$ the pairs $(c_j^-,c_{j+1}^-)$ and $(c_j^+,c_{j+1}^+)$ of clause agents are also enemies. Each clause agent~$c_j^+$ or $c_j^-$ has three variable agents as their friends, namely the ones corresponding to literals present in the clause $C_j$. All remaining relations are neutral. To finalize the construction, we set $\numCoals = 2$.
    
    For correctness, assume first that $\alpha$ is a valid assignment for $\varphi$. 
    We construct a coalition structure $\pttn$ so that $\pttn_1 = \{ r^+ \} \cup \{ c_j^+ \colon j\in[m] \} \cup \{ a_i \colon \alpha( x_i ) = \texttt{true} \}$ and $\pttn_2 = \agents \setminus \pttn_2$. Clearly, both guard agents have only neutral agents in their coalition, so the partition is IR for them. Variable agents have no enemies, so any allocation is IR for these. Finally, each clause agent~$c$ is in the coalition with at least one enemy. However, since $\alpha$ is a valid assignment, $c$ has at least one friend in the same coalition---namely, the variable agent~$a_i$ corresponding to the variable~$x_i$ that satisfies $C_j$ if $c=c^+_j$ and the one that does not satisfy $C_j$ if $c=c^-_j$. That is, $\pttn$ is IR for all agents, %
    so $\mathcal{J}$ is a \Yes-instance.

    In the opposite direction, let $\pttn=(\pttn_1,\pttn_2)$ be an IR coalition structure for~$\mathcal{J}$. We first observe that $r^+$ and $r^-$ cannot be members of the same coalition, since they do not have friends and are mutual enemies. By the same argument, no clause agent~$c^-_j$ can be in the same coalition with $r^+$, and no clause agent~$c^+_j$ can be in the same coalition as $r^-$. Hence, without loss of generality, we can assume that $\{r^+,c_1^+,\ldots,c_m^+\} \subseteq \pttn_1$ and $\{r^-,c_1^-,\ldots,c_m^-\}\subseteq \pttn_2$. We construct an assignment $\alpha$ by setting $\alpha(x_i) = \texttt{true}$ if and only if $a_i\in\pttn_1$ for every $i\in[n]$. For contradiction, assume that $\alpha$ is not valid for $\varphi$, i.e., there is a clause $C_j$ with either all or none of its literals being true. In the former case,  all friends of~$c_j^-$ are in~$\pttn_1$. However, $c_j^-$ has at least one enemy in $\pttn_2$, which contradicts the individual rationality of~$\pttn$. In the latter case, we apply the same argument on~$c_j^+$. Consequently, $\alpha$ is a valid truth assignment, finishing the proof for \IRkfo. 
    For \IRBSCfo, we simply add $n$ dummy agents $d_1,\ldots,d_n$ that are neutrals from the perspective of all agents. 
}

\toappendix{
Let us remark that padding with neutral dummy agents, as in the last step of the previous proof, does not yield a general reduction from \IRkfo\ to \IRBSCfo: in a balanced solution of the padded instance, some coalitions may consist solely of dummy agents, and removing them then leaves fewer than~$\numCoals$ non-empty coalitions. The argument works above because the guard agents force both coalitions to contain original agents.}

Next, we show that aiming for small coalitions does not yield tractability: our problems are hard already if we require each coalition to be of size at most~$3$. 
The reduction is from the \NPh \probName{$\{P_3,K_3,K_{1,3}\}$-Decomposition} problem, which is known to be \NPh even on graphs of maximum degree $3$~\cite{BulteauFLRR2021}. Given a graph $H=(W,F)$ with~$|F| = 3m$, this problem asks if $F$ can be partitioned into sets $F_1,\ldots,F_m$ of size exactly~$3$ so that each $F_i$ forms a 3-edge path~($P_3$), a triangle~($K_3$), or a star with three edges~($K_{1,3}$).
The main idea of our reduction is to introduce an agent corresponding to each edge of the input graph~$H$, with two agents being friends if the corresponding edges share an end-vertex and being enemies otherwise. It is not hard to see that in such a game, a coalition of size~$3$ is IR iff the corresponding edges in~$H \simeq P_3,K_3,K_{1,3}$.     

\begin{restatable}[\linkproof{thm:FO:BSC:ubAtMostThree:NPc}]{theorem}{thmFOBSCubAtMostThreeNPc}
\label{thm:FO:BSC:ubAtMostThree:NPc}
    \IRBSCfo is \NPc 
    even if $\ub(j) = 3$ for $\forall j\in[\numCoals]$, preferences are symmetric, $|\friends_a| \in \{3,4\}$ for $\forall a\in\agents$, and there are no neutrals.
\end{restatable}

\prooftoappendix{thm:FO:BSC:ubAtMostThree:NPc}{\thmFOBSCubAtMostThreeNPc*}{
    We reduce from the \probName{$\{P_3,K_3,K_{1,3}\}$-Decomposition} problem, which is known to be \NPh even on graphs of maximum degree $3$~\cite{BulteauFLRR2021}. In this problem, we are given a graph $H=(W,F)$, with $|F| = 3m$, and the goal is to partition~$F$ into sets $F_1,\ldots,F_m$ of size exactly three so that each $F_i$ forms a 3-edge path~($P_3$), a triangle~($K_3$), or a star with three edges~($K_{1,3}$).

    Given an instance $\mathcal{I}$ of the \probName{$\{P_3,K_3,K_{1,3}\}$-Decomposition} problem, we construct an equivalent instance $\mathcal{J}$ of \IRBSCfo as follows. For every edge $e_i\in F$, we create an agent $a_i$. Two agents $a_i$ and $a_j$ are  friends if and only if $e_i \cap e_j \not= \emptyset$. Since the degree of any vertex $v\in W$ is at most three, each agent has at most four friends. Additionally, we set $\enemies_a = \agents\setminus(\friends_a \cup \{a\})$ for each~$a\in\agents$. That is, there are no two agents that are neutral to each other. To finalize the construction, we set $\numCoals = m$ and $\lb(j) = \ub(j) = 3$ for every $j\in[\numCoals]$.

    For correctness, assume that $\mathcal{I}$ is a yes-instance and $F_1,\ldots,F_m$ is a solution partition of $F$. We construct $\pttn$ by setting $\pttn_j = \{ a_i \colon e_i \in F_j \}$ for every $j\in[m]$. As each $F_j$ forms $P_3$, $K_3$, or $K_{1,3}$, each coalition is clearly of size $3$. Additionally, the form of each $F_j$ implies that every agent $a_i$ has at least one friend in $\pttn(a_i)$. Thus, $\pttn$ is IR.

    In the opposite direction, let $\mathcal{J}$ be a yes-instance and $\pttn$ be a solution structure~$\pi$; in particular, each coalition in~$\pi$ has size~$3$. Set $F_j = \{ e_i \colon a_i\in \pttn_j \}$.  For the sake of contradiction, assume that there exists $F_j = \{ e_{i_1}, e_{i_2}, e_{i_3} \}$ that does not form $P_3$, $K_3$, or $K_{1,3}$. Consequently, there is an edge in $F_j$, say $e_{i_1}$, such that $e_{i_1} \cap e_{i_2} = e_{i_1} \cap e_{i_3} = \emptyset$. Hence,  agent $a_{i_1}$ has no friend in $\pttn(a_{i_1})$ but has two enemies, which contradicts that $\pttn$ is IR. Therefore, each $F_j$ forms a $P_3$, a $K_3$, or a $K_{1,3}$ in~$H$, which completes the proof.
}

Contrasting \Cref{thm:FO:BSC:ubAtMostThree:NPc}, we show that if each coalition must be either a pair of agents or a singleton, then our problems are solvable efficiently.
Our algorithm relies on the observation that a coalition formed by two agents is IR if and only if there is no enmity between these two agents (in any direction). By searching for a matching of an appropriate size in the graph where two agents are connected by an edge if and only if neither considers the other an enemy, we can solve our instance in the same time necessary for finding a maximum matching.
\begin{restatable}[\linkproof{thm:FO:SCC:ubAtMostTwo:poly}]{theorem}{thmFOSCCubAtMostTwopoly}
\label{thm:FO:SCC:ubAtMostTwo:poly}
    If $\ub(j) \leq 2$ for every $j\in[\numCoals]$, then \IRSCCfo{} $\in \cc{P}$.%
\end{restatable}
\prooftoappendix{thm:FO:SCC:ubAtMostTwo:poly}{\thmFOSCCubAtMostTwopoly*}{
    Before we present our algorithm, we introduce an additional notation. First, we use $G^{\bar{\enemies}}$ to denote an auxiliary graph on the set of agents such that two agents $a$ and $b$ are connected with an edge if and only if neither $a \in \enemies_b$ nor $b \in \enemies_a$. Next, we use $n_{\nicefrac{u}{\ell}}$ to denote the number of all coalitions $j\in[\numCoals]$ such that $\ub(j) = u$ and $\lb(j) = \ell$. Without loss of generality, we assume $n_{\nicefrac{0}{0}} = 0$, as such coalitions can be removed without changing the solution of the game.

    We try all candidate values for the number~$k_2$ of coalitions of size~$2$ in a hypothetical solution~$\pttn$; note that we must have $n_{\nicefrac{2}{2}} \leq  k_2 \leq n_{\nicefrac{2}{2}}+n_{\nicefrac{2}{1}}+n_{\nicefrac{2}{0}}$. Let us remark that fixing $k_2$ to the minimum feasible value~$n_{\nicefrac{2}{2}}$ does not suffice: if there are more remaining agents than remaining coalitions, additional coalitions of size two are needed to accommodate all agents.
    Next, we compute a maximum matching~$M$ in~$G^{\bar\enemies}$.
    Since each coalition of size~$2$ must be a matching in $G^{\bar\enemies}$, we must have $k_2 \leq |M|$; otherwise, we can reject the current value (and continue with the next candidate value for~$k_2$).
    Thus, we may assign $k_2$ arbitrary edges from~$M$ to~$k_2$ coalitions with upper bound~$2$, giving precedence to those with lower bound~$2$ (as $k_2 \geq n_{\nicefrac{2}{2}}$, all such coalitions are now filled up), and then to those with lower bound~$1$. 
    
    Let $k_1$ denote the remaining still empty coalitions with lower bound~$1$. Clearly, we must have $|\agents|-2k_2 \geq k_1$ as otherwise, there are not enough agents to satisfy every lower bound, in which case we can reject. Moreover, if our guesses are correct, then the remaining $k-k_2$ coalitions can each contain at most one agent; this means that we must have $|\agents|-2k_2  \leq k-k_2$. However, if both of these conditions hold, i.e. 
    \begin{equation}
    \label{eq:numberOfSingletons}
    k_1 \leq |\agents|-2k_2 \leq k-k_2,
    \end{equation}
    then we can distribute the remaining $|\agents|-2k_2 $ agents among the yet empty coalitions so that each coalition~$j \in [k]$ with $\lb(j)=1$ contains an agent, 
    and no coalition (other than those containing the $k_2$ pairs from~$M$) contains more than one agent. This algorithm necessitates to check for each $k_2 \in \{n_{\nicefrac{2}{2}},\dots,n_{\nicefrac{2}{2}}+n_{\nicefrac{2}{1}}+n_{\nicefrac{2}{0}}\}$ whether~\Cref{eq:numberOfSingletons} holds. Therefore, the bottleneck in the algorithm is the maximum matching computation, yielding a running time of~$\Oh{\sqrt{|\agents|}\cdot|E(G^{\bar{\enemies}})|}$ using the algorithm by Micali and Vazirani~\citeapp{MicaliV1980} for computing~$M$. 
}

\subsection{Small Number of Enemies}
First, observe that if no agent has enemies, then any coalition structure is IR. Hence, we only need to check whether the size constraints can be respected.
In fact, for this, we only need to check whether 
$\numAgents \geq \sum_{i \in [k]} \lb(i)$ and $\numAgents \leq \sum_{i \in [k]} \ub(i)$, yielding a linear-time algorithm. 

\begin{restatable}[\linkproof{thm:FO:SCC:noEnemies:poly}]{proposition}{thmFOSCCnoEnemiespoly}
\label{thm:FO:SCC:noEnemies:poly}
    If $|\enemies_a| = 0$ for every $a\in\agents$, then 
    \IRSCCfo $\in \cc{P}$.%
\end{restatable}
\prooftoappendix{thm:FO:SCC:noEnemies:poly}{\thmFOSCCnoEnemiespoly*}{
    First, we need to check if $\numAgents \geq \sum_{i \in [k]} \lb(i)$ and $\numAgents \leq \sum_{i \in [k]} \ub(i)$, as otherwise the size constraints cannot be respected by any coalition structure. If so, then it is not hard to see that the instance is a \Yes-instance:
    since any coalition structure is IR (because there are no enmities), we can first fill up each coalition~$i$ up to its lower bound~$\lb(i)$, and then distribute the remaining agents in an arbitrary way as long as no upper bound is exceeded.    
}

If each agent has only one enemy, then we can create a balanced coalition structure (i.e., where all coalition sizes differ by at most~$1$) 
using a simple and fast algorithm: It suffices to place each agent that has some enemy into a coalition not shared by their unique mutual enemy. To achieve this, we fix an ordering of the agents so that agents with at least one enemy precede those without enemies, and each agent is next to their enemy (if they have one). Applying the round robin method based on this ordering to place the agents into~$k$ coalitions, we obtain an IR coalition structure.
\begin{restatable}[\linkproof{thm:FO:k:oneEnemy:poly}]{theorem}{thmFOkoneEnemypoly}   
\label{thm:FO:k:oneEnemy:poly}
\label{thm:FO:BSC:oneEnemy:poly}
    If $|\enemies_a| \leq 1$ for every $a\in\agents$ and enmities are symmetric, then 
    \IRBSCfo and \IRkfo can be solved in polynomial time. 
\end{restatable}
\prooftoappendix{thm:FO:k:oneEnemy:poly}{\thmFOkoneEnemypoly*}{
    If $\numCoals=1$, we can decide the instance directly using \Cref{thm:FO:SCC:grandCoal:poly}. Hence, we assume $\numCoals \geq 2$. We fix an order $a_1,\ldots,a_\numAgents$ of agents so that agents with at least one enemy are first and each of them is next to their enemy. Only after these agents follow the agents without enemies.

    Now, we put the $i$-th agent in the $((i\bmod \numCoals)+1)$-th coalition. Since $\numCoals\geq 2$, no pair of enemies ends up in the same coalition. Consequently, the coalition structure is IR. Moreover, it is easy to see that each coalition is non-empty and the coalitions differ in size by at most $1$.
}

\Cref{thm:FO:BSC:oneEnemy:poly} is contrasted by our next result, which shows that both \IRBSCfo\ and \IRkfo\ become intractable once each agent has four enemies, even if the number of coalitions is only~$k=3$. 
This result follows from the \NPhness of the classic $3$-coloring problem and its equitable variant, where the sizes of the color classes have to differ by at most~$1$; both of these problems are known to be \NPc even on $4$-regular graphs~\cite{FurmanczykM2022}. 
Our reduction is very simple: each vertex is represented by an agent with no friends, and two agents are enemies if and only if the corresponding vertices are adjacent in the input graph.

\begin{restatable}[\linkproof{thm:FO:k:fourEnemies:NPh}]{theorem}{thmFOkfourEnemiesNPh}
\label{thm:FO:k:fourEnemies:NPh}
    \IRBSCfo and \IRkfo are \NPc even if preferences are symmetric,  $\numCoals=3$, and $|\enemies_a|=4$ for every $a\in\agents$. 
\end{restatable}
\prooftoappendix{thm:FO:k:fourEnemies:NPh}{\thmFOkfourEnemiesNPh*}{
    We show hardness of \IRkfo by a simple reduction from the \probName{$3$-Coloring} problem, which is \NPc even if the underlying graph $H$ is $4$-regular~\citeapp{GareyJS1976}. Moreover, we assume that $H$ is not $2$-colorable.

    We set $G=H$ and color all edges in red. That is, the agents have only enemies and no friends. Now, given a coloring of $H$, each color class induces an independent set, so if we use the same partition, we have an IR partition with three non-empty coalitions. In the opposite direction, each coalition of an IR coalition structure must induce an independent set, as otherwise it contains two agents that are mutual enemies. That is, coloring the agents according to the part they belong to in~$\pttn$ is a proper coloring of the original graph.

    The hardness for \IRBSCfo is identical; we just start with the \probName{Equitable $3$-Coloring} problem, which asks for a coloring where every color class is of the same size and is \NPc under the same restrictions~\cite{FurmanczykM2022}.
}

For the variant where the coalition structure is not required to be balanced but the coalition sizes are still fixed, the problem remains hard even if each agent has exactly three enemies;
to show this, we use a reduction  from  the \NPc \probName{Semi-Equitable $3$-Coloring} problem~\cite{FurmanczykK2016} where the color classes of an $n$-vertex input graph need to be of size~$4n/10$, $3n/10$, and $3n/10$.

\begin{restatable}[\linkproof{thm:FO:FSC:threeEnemies:NPh}]{theorem}{thmFOFSCthreeEnemiesNPh}
\label{thm:FO:FSC:threeEnemies:NPh}
    \IRFSCfo is \NPc even if  preferences are symmetric,  $\numCoals=3$, and $|\enemies_a| = 3$ for every $a\in\agents$.
\end{restatable}
\prooftoappendix{thm:FO:FSC:threeEnemies:NPh}{\thmFOFSCthreeEnemiesNPh*}{
    The reduction is the same as in \Cref{thm:FO:k:fourEnemies:NPh}, with the difference that we reduce from the \NPc \probName{Semi-Equitable $3$-Coloring} problem~\cite{FurmanczykK2016}. In this problem, given a cubic graph~$H$, the task is to find a proper $3$-coloring of~$H$ such that the sizes of the color classes are $4n/10$, $3n/10$, and $3n/10$. Hence, we take the same graph, color all edges red, and set $\lb(1)=\ub(1)= 4n/10$ and $\lb(j)=\ub(j)=3n/10$ for $j\in\{2,3\}$. The correctness then follows from the arguments of \Cref{thm:FO:k:fourEnemies:NPh}.
}

\subsection{Small Number of Recluses}
Notice that under friend-oriented preferences, agents without friends are usually the source of computational intractability, since they must be placed in a coalition that contains none of their enemies, and thus we may need to seek coalitions that form independent sets in the enemy graph---an idea utilized in our reduction from \probName{3-Coloring} and its variants in the proofs of
\Cref{thm:FO:k:fourEnemies:NPh,thm:FO:FSC:threeEnemies:NPh}. Hence, it makes sense to study how the number of such agents influences the complexity of our problems, motivating the following definition.

\begin{definition}
    An agent $a\in\agents$ is called a \emph{recluse} if $\friends_a = \emptyset$ and $|\enemies_a| \geq 1$.
\end{definition}

Notice that with upper and lower bounds on the number of coalition sizes, \Cref{thm:FO:BSC:ubAtMostThree:NPc}
implies that even symmetric instances with no recluses are intractable. For \IRkfo, i.e., when only the number of coalitions is given, \Cref{thm:FO:BSC:twoCoals:NPh} shows that even two recluses suffice to cause \NPhness. By contrast, \IRkfo can be solved if preferences are symmetric and there are no recluses.

\begin{theorem}\label{thm:FO:k:symmetric:atLeastOneFriend:poly}
    If %
    there are no recluses, every symmetric \IRkfo $\Gamma$ admits an individually rational coalition structure which can be found in polynomial time.
\end{theorem}
\begin{proof}
    Recall that $\numCoals \leq \numAgents$. We first identify all connected components $K_1,\ldots,K_\ell$ of the friendship graph $G^\friends$. %
    Due to the absence of recluses,  each~$K_j$ is either of size at least two or contains a single agent that has neither friends nor enemies.

    If $\numCoals \leq  \ell$, then we set $\pttn_j = V(K_j)$ for each $j\in[\numCoals-1]$ and $\pttn_\numCoals = \bigcup_{j=\numCoals}^\ell 
    V(K_j)$. 
    Then $\pttn$ has $k$ non-empty coalitions and is IR due to our previous observation.

    If $\numCoals > \ell$, then for every component $K_j$, $j\in[\ell]$, we compute a spanning tree~$T_j$. 
    We start by setting $\pttn_j = V(K_j)$ for every $j\in[\ell]$. This is IR as shown above, but coalitions~$\pi_j$ for $j>\ell$ are empty. Hence, we do the following: let $\pttn_j$ be a non-empty coalition of size at least two and $\pttn_{j'}$ be an empty coalition. We find agent $a\in\pttn_j$ that is a leaf of the spanning tree~$T_j$; we move~$a$ from~$\pttn_j$ to~$\pttn_{j'}$. %
    Now, agent $a$ is in a singleton coalition, so $\pttn_{j'}$ is IR for~$a$. Since we removed a leaf from the tree of~$T_j$, the coalition~$\pttn_j \setminus \{a\}$ remains also IR, since each agent of~$\pttn_j\setminus\{a\}$ still has at least one friend in~$\pttn_j\setminus\{a\}$---their neighbor in $T_j\setminus\{a\}$. The procedure terminates once all~$k$ coalitions are non-empty. The invariant of individual rationality is satisfied after each transfer.%
\end{proof}

Unfortunately, the algorithm of Theorem~\ref{thm:FO:k:symmetric:atLeastOneFriend:poly} cannot be extended to the asymmetric case, as our next result shows.

\begin{restatable}[\linkproof{thm:FO:k:asymmetric:atLeastOneFriend:NPh}]{theorem}{thmFOkasymmetricatLeastOneFriendNPh}
\label{thm:FO:k:asymmetric:atLeastOneFriend:NPh}
\IRkfo is \NPc even if there are no recluses.
\end{restatable}
\begin{proofsketch}
    We reduce from \probName{3-Coloring}. Let $H$ be the input graph. We construct an instance $\Gamma$ of \IRkfo as follows.
    Let us set~$K=|V(H)|^2$. For each vertex $u \in V(H)$, we introduce a \emph{vertex-agent}~$a_u$ and $K-1$ additional dummy agents; these $K$ agents together will form a directed cycle of length~$K$ in the friendship graph~$G^\friends$. In fact, $G^\friends$ will be the disjoint union of these $|V(H)|$ cycles, which immediately implies that no recluses are contained in our instance, as each agent has exactly one friend. 
    Besides these friendships, we let each dummy agent consider every agent other than their unique friend and themselves an enemy. Each vertex agent considers all dummy agents other than their unique friend as an enemy. Finally, two vertex agents $a_u$ and~$a_{v}$ consider each other mutual enemies if $u$ and~$v$ are adjacent in~$H$. Finally, we set %
    $k=(K-1) \cdot |V(H)| +3$. It is not hard to verify that $H$ is 3-colorable if and only if $\Gamma$ admits an IR coalition structure with exactly $k$ non-empty coalitions.
\end{proofsketch}
\prooftoappendix{thm:FO:k:asymmetric:atLeastOneFriend:NPh}{\thmFOkasymmetricatLeastOneFriendNPh*}{
    We reduce from \probName{3-Coloring}. Let $H$ be the input graph. We construct an instance $\Gamma$ of \IRkfo as follows.
    Let us set~$K=|V(H)|^2$. For each vertex $u \in V(H)$, we introduce a \emph{vertex-agent}~$a_u$ and $K-1$ additional dummy agents; these $K$ agents together will form a directed cycle of length~$K$ in the friendship graph~$G^\friends$. In fact, $G^\friends$ will be the disjoint union of these $|V(H)|$ cycles, which immediately implies that no recluses are contained in our instance, as each agent has exactly one friend. 
    Besides these friendships, we let each dummy agent consider every agent other than their unique friend and themselves an enemy. Each vertex agent considers all dummy agents other than their unique friend as an enemy. Finally, two vertex agents $a_u$ and~$a_{v}$ consider each other mutual enemies if $u$ and~$v$ are adjacent in~$H$. Finally, we set %
    $k=(K-1) \cdot |V(H)| +3$. It is not hard to verify that $H$ is 3-colorable if and only if $\Gamma$ admits an IR coalition structure with exactly $k$ non-empty coalitions.
    
    First, assume that $H$ is 3-colorable. Create $3$ coalitions by grouping the vertex agents according to the $3$ color classes, and put every dummy into a singleton coalition. No vertex agent~$a_u$ has an enemy in their coalition, since the vertices corresponding to the coalition containing~$a_u$ form an independent set in~$H$. Therefore, the obtained coalition structure is IR and yields exactly $k$ coalitions.

    Assume now that $\Gamma$ admits an IR coalition structure $\pttn$ with exactly $k$ non-empty coalitions. Note that since each dummy vertex hates every agent other than their unique friend, we know that they either form a singleton or must be contained in the same coalition as their friend. This implies that each cycle in~$G^\friends$ either must be contained completely in some coalition of~$\pttn$ or all of its dummy agents form singletons in~$\pttn$. If not all dummies are contained in singletons, then the number of coalitions is at most $(|V(H)|-1)K+1<k$. Thus, all dummies form singletons in~$\pttn$. However, this yields exactly $(K-1)|V(H)|$ singletons in~$\pttn$, so all vertex agents must be grouped into three coalitions. As there are no friendships between vertex agents, these coalitions cannot contain a pair of enemies---that is, the corresponding vertices in~$H$ must form an independent set.%
}

\noindent\Cref{thm:FO:k:symmetric:atLeastOneFriend:poly,thm:FO:k:asymmetric:atLeastOneFriend:NPh} leave open only the symmetric case with exactly one recluse.

\subsection{Restricting the Enmity Graph}
Next, we examine how the structure of the enmity graph influences the tractability of our problems. We consider the vertex cover number of~$G^\enemies$, denoted by~$\vc(G^\enemies)$. We start with a simple algorithm for the case  
when $\vc(G^\enemies) \leq 1$: the main idea is to place the unique problematic agent---the one incident to all edges in~$G^\enemies$---into a singleton.%

\begin{restatable}[\linkproof{thm:FO:k:enmityVcAtMostOne:poly}]{theorem}{thmFOkenmityVcAtMostOnepoly}
\label{thm:FO:k:enmityVcAtMostOne:poly}
    If $\vc(G^\enemies) \leq 1$, 
    then \IRkfo can be solved in linear time.
\end{restatable}
\prooftoappendix{thm:FO:k:enmityVcAtMostOne:poly}{\thmFOkenmityVcAtMostOnepoly*}{
    If there are no enmities, then we can solve the instance by \Cref{thm:FO:SCC:noEnemies:poly}.
    So we assume $\vc(G^\enemies) = 1$, and let $c$ be the agent covering all edges in~$G^\enemies$.
    
    If $\numCoals=1$, we can solve the instance in linear time by \Cref{thm:FO:SCC:grandCoal:poly}. 
    Hence, $\numCoals\geq 2$. 
    In this case, however, we can put agent~$c$ alone into coalition $\pttn_1$ and split the remaining agents arbitrarily between the remaining coalitions. As the only enmities are from/towards $c$, this is clearly IR.
}

For symmetric instances, in the case $\vc(G^\enemies) \leq 1$, a polynomial-time algorithm based on a more careful case analysis can solve the \IRSCCfo problem.

\begin{theorem}
\label{thm:SSC:k:symmetric:enmityVcAtMostOne:poly}
    If $\vc(G^\enemies) \leq 1$, 
    then \IRSCCfo for symmetric instances can be solved in polynomial time.
\end{theorem}

\begin{proof}
    If there are no enmities, then we can solve the instance by \Cref{thm:FO:SCC:noEnemies:poly}.
    So we assume $\vc(G^\enemies) = 1$, and let $c$ be the agent covering all edges in~$G^\enemies$. 
    We start by first guessing the coalition~$j$ containing~$c$ and its size~$s_j \in [\numAgents]$ in a solution~$\pi$. Next, we guess whether $c$ has an enemy in~$\pi_j=\pi(c)$.

    If $c$ has no enemies in~$\pi_j$, then we can put an arbitrary set of~$s_j-1$ agents into a coalition~$C$ together with~$c$ from among $\agents \setminus \enemies_c$, and set $C$ is the $j$-th coalition; if $|\agents \setminus \enemies_c| <s_j$,  we can reject the current set of guesses. We can then distribute the remaining agents across the remaining coalitions arbitrarily in a way that respects the lower and upper bounds, if this is possible, since there are no enmities between the remaining agents (see the method in \Cref{thm:FO:SCC:noEnemies:poly}).

    If $c$ has some enemies in~$\pi_j$, then we also guess a friend~$c'$ of~$c$ in~$\pi_j$. Next, 
    we start adding agents to a coalition~$C$, which initially contains~$c$ and~$c'$, as follows: as long as $|C|<s_j$, we try to add an agent to~$C$ that either does not consider~$c$ an enemy or considers some agent in~$C$ a friend. If this procedure stops with~$|C|<s_j$, then consider %
    the set~$X$ of those remaining agents---all of them enemies of~$c$---that have at least one friend. We aim to find 
    a subset~$X' \subseteq X$ of~$s'_j$ agents that is \emph{valid}, that is, each of them has at least one friend in~$X'$. Adding a valid set of size~$s'_j$ to~$C$ yields an IR coalition of size~$s_j$.
    
    Observe that if $s'_j=s_j-|C| > |X|$, then no IR coalition exists of size~$s_j$ containing~$c$, since every enemy of~$c$ in~$\pi_j$ must have a friend in~$\pi_j$ too; thus, we can reject.
    Otherwise, it is possible to find a valid set~$X'$ unless (i) $s'_j=1$, or (ii) $s'_j$ is odd, but the subgraph of~$G^\friends$ induced by~$X$ is the union of independent edges (i.e., a matching). In these cases, if $C=\{c,c'\}$, then no IR coalition of size~$s_j$ contains~$c$, because each edge in~$G^\friends[X]$ must have either 0 or~$2$ of its endpoints in such a coalition. If $|C|>2$, then it suffices to remove the agent added at the last step to~$C$, and add to~$C$ a valid set of~$s'_j+1$ agents from~$X$, obtaining an IR coalition of size~$s_j$ containing~$c$.
    Since there are no enmities between the remaining agents, we can distribute them respecting the upper and lower bounds of the coalitions (or reject if this is not possible). 
    There are at most~$\numAgents^3$ guesses, and each of them can be verified in linear time.\iflong\footnote{In fact, with some care, the algorithm can work without guessing~$|\pi_j|$ in advance, and instead determining all possible sizes that an IR coalition containing~$c$ can have, shaving off a factor of~$\numAgents$ from the running time.}\fi{}
\end{proof}

Unfortunately, the algorithm in \Cref{thm:SSC:k:symmetric:enmityVcAtMostOne:poly} cannot be generalized to the asymmetric case, as shown by the following strong intractability result.

\begin{restatable}[\linkproof{thm:BSC:k:enmityVcAtMostOne:NPh}]{theorem}{thmBSCkenmityVcAtMostOneNPh}
\label{thm:BSC:k:enmityVcAtMostOne:NPh}
    \IRBSCfo is \NPc even if $\vc(G^\enemies)=1$, $\numCoals=2$,
    and $|\enemies_a|\leq 1$ and $|\friends_a| \leq 2$ for each $a \in \agents$.
\end{restatable}
\prooftoappendix{thm:BSC:k:enmityVcAtMostOne:NPh}{\thmBSCkenmityVcAtMostOneNPh*}{
    We reduce from the \probName{Independent Set} problem on regular graphs. Let the input of this problem be~$(H,q)$ for some graph~$H$ in which all vertices have degree~$r$ for some~$r \in \mathbb{N}$. We construct an instance~$\mathcal{I}$ of~\IRBSCfo as follows. 

    Let $n=|V(H)|$; then $m=|E(H)|=\nicefrac{rn}{2}$.
    Let us define integers
    \begin{align*}
        \ell &= rn+1,  
         & 
         L&=n^4+\ell(n-q)+(m-qr),    \\
         s&=(\ell+r)q,
         &
         L'&=n^4+s+1.
    \end{align*}
    
    For each vertex~$v \in V(H)$, we create $\ell$ \emph{vertex agents}~$a_v^0,\dots,a_v^{\ell-1}$, 
    and for each edge~$e \in E(H)$, we create an \emph{edge agent $a_e$}. 
    Next, we create two sets of dummy agents: $D=\{d_0,\dots,d_{L-1}\}$ and $D=\{d'_0,\dots,d'_{L'-1}\}$. Finally, we add a special agent~$c$. We let~$\agents$ denote the set of all agents. 

    The relations of the agents are defined as follows. All agents consider~$c$ as an enemy, but there are no other enmities, meaning that $\{c\}$ is a vertex cover for the enmity graph~$G^\enemies$.
    The friendship graph is constructed as follows: for each vertex~$v \in V(H)$, the corresponding vertex agents form a directed cycle, i.e., each agent~$v_a^i$ considers only~$v_a^{(i +1) \bmod \ell}$ a friend, but no one else. 
    Similarly, the dummy agents in~$D$ form a friendship cycle: $d_i$ considers only~$d_{(i+1) \bmod L}$ a friend. Agents in~$D'$ have no friends. 
    Lastly, each edge agent~$a_e$ corresponding to an edge~$e$ with endpoints~$u$ and~$v$ in~$H$ considers~$a_u^0$ and~$a_v^0$ as friends. To finish the construction of our instance~$\mathcal{I}$, we set the number of desired coalitions to be~$\numCoals=2$.

    Suppose that $\mathcal{I}$ admits a solution~$(\pttn_1,\pttn_2)$; we may assume $c \in \pttn_1$ by symmetry. Note first that agents in~$D'$ cannot be contained in~$\pttn_1$, as $c$ is an enemy of each of them, and they have no friends. 
    Second, observe that for each set~$A^\circ$ of agents that forms a cycle in~$G^\friends$, either $A^\circ \cap \pttn_1=\emptyset$ or $A^\circ \subseteq \pttn_1$, because any agent~$a \in A^\circ \cap \pttn_1$ must have their unique friend in~$\pttn_1$ (by $c \in \enemies_a \cap \pttn_1$)---i.e., the next agent on the cycle.
    
    Note that the number of agents is $\ell n+m+L+L'+1$, so both coalition must have size~\[|\agents|=\frac{\ell n+m+L+L'+1}{2}=
    n^4+\ell n + m+1,\] 
    which implies that $D \subseteq \pttn_1$, as otherwise $|\pttn_2|\geq |D|+|D'|=L+L'>2n^4>\nicefrac{|\agents|}{2}$. 
    Due to $\nicefrac{|\agents|}{2}=L+s+1$, it follows that $\pi_1$ contains~$D \cup \{c\}$ and a set of exactly~$s$ vertex and edge agents. Since the set of vertex agents forms $n$ disjoint cycles of length~$\ell$ in~$G^\friends$, by $\ell>m$ we know that $\pttn_1$ must contain exactly~$q$ such cycles, corresponding to a set~$S$ of~$q$ vertices in~$H$. Moreover, $\pttn_1$ must also contain $qr$ edge agents. Note that each such edge agent~$a_e$ must have a friend in~$\pttn_1$, i.e., $e$ must have an endpoint in~$S$. However, as  $H$ is $r$-regular and $|S|=q$, there can be at most~$qr$ edges incident to some vertex in~$S$, and equality happens only if~$S$ is an independent set. Thus, $S$ is an independent set of size~$q$, as required.      

    For the other direction, it is straightforward to verify that given an independent set~$S$ of size~$q$ in~$H$, the $\ell q$  vertex agents corresponding to the vertices in~$H$, the $r q$ edge agents corresponding to the edges incident to some vertex in~$H$, and the agents in~$D \cup \{c\}$ form an coalition of size~$\nicefrac{|\agents|}{2}$ that is IR. Since $c$ covers all edges in~$G^\enemies$, all remaining~$\nicefrac{|\agents|}{2}$ agents can be put together into another IR coalition, proving yielding a solution for~$\mathcal{I}$. 
}

If preferences are symmetric, we can generalize the algorithm of \Cref{thm:FO:k:enmityVcAtMostOne:poly} for a vertex cover of size~$3$, using an extensive case analysis. 
\begin{restatable}[\linkproof{thm:FO:k:symmetric:enmityVcAtMostThree:poly}]{theorem}{thmFOksymmetricenmityVcAtMostThreepoly}
\label{thm:FO:k:symmetric:enmityVcAtMostThree:poly}
    If $\vc(G^\enemies) \leq 3$ and preferences are symmetric, \IRkfo can be solved in polynomial time. 
\end{restatable}
\prooftoappendix{thm:FO:k:symmetric:enmityVcAtMostThree:poly}{\thmFOksymmetricenmityVcAtMostThreepoly*}{
    We may assume $\numCoals\geq 2$. Let $C$ be a vertex cover of~$G^\enemies$ of minimal size. Let us assume that there exists a solution~$\pi$ for our instance of \IRkfo.
    
    We first guess the partitioning of~$C$ induced by~$\pi$, and we further guess those agents in~$C$ that are \emph{lonely} in~$\pi$, meaning that they have no friend in their coalition in~$\pi$. For each agent in~$c \in C$ that is not lonely, we also guess a friend~$a_c$ of~$c$ in~$\pi(c)$. %
    Let $\wt{\pi}$ denote the obtained partitioning of the set~$\wt{C}$ of all of these agents, i.e., of all agents in~$C$, and the friend~$a_c$ for each lonely agent~$c$ in~$C$.
    Note that each agent~$a \in \wt{C}$ who has an enemy in~$\wt{\pi}(a)$ also has a friend in~$\wt{\pi}(a)$, if our guesses are correct.
    We distinguish between the following cases. 

    \proofcase{A}{$\numCoals> |\wt{\pi}|$.}
    Assuming correct guesses, we know that  $|\agents \setminus \wt{C}| \geq \numCoals -|\pi|$ must hold, as there are exactly $\numCoals -|\pi|$ coalitions containing no agents of~$C$ in~$\pi$, and these coalitions contain a subset of the agents in~$\agents \setminus  \wt{C}$.
    Therefore, we can distribute all agents in $\agents \setminus \wt{C}$ across $\numCoals-|\pi|$ newly created coalitions in an arbitrary manner. 
    Together with $\wt{\pi}$, this yields an IR coalition structure with exactly $\numCoals$ coalitions. 

    \proofcase{B}{$\numCoals=|\wt{\pi}|$ and all agents in~$C$ are lonely.} 
    In this case, we iterate over each of the remaining agents, putting each of them into some (arbitrarily chosen) coalition where they do not have any enemies. Since no coalition in~$\pi$ can contain a pair of enemies (as all agents in~$C$ are lonely), and any enemy of some agent in~$\agents \setminus C$ must be in~$C$, we can place each agent in this manner. In the resulting coalition structure, no agent shares a coalition with an enemy, so individual rationality is satisfied.

    \proofcase{C}{$\numCoals=|\wt{\pi}|$ and there is an agent~$c^\star \in C$ that is not lonely.}
    If $\pi(c^\star)$ contains some lonely agent~$c$ in~$C$, then 
    we know that $\numCoals=2$  (because  $\numCoals=|\wt{\pi}| \leq |C|-1\leq 2$), and the other coalition in~$\pi$ (not $\pi(c)=\pi(c^\star)$) contains exactly one agent of~$C$ which we will denote by~$c'$. Note that 
    no enemy of~$c$ can be contained in $\pi(c^\star)$, so we add all of them to~$\wt{\pi}(c')$. 

    Regardless of whether $\pi(c^\star)$ contains some lonely agent of~$C$ or not, we proceed as follows.
    Let $S$ denote the set of all agents placed into some coalition of~$\wt{\pi}$ other than~$\wt{\pi}(c^\star)$. 
    We start an iteration as follows. 
    While there exists some unplaced agent~$a \notin \wt{\pi}(c^\star) \cup S$ that has an enemy in~$\wt{\pi}(c^\star)$ but has friends only in~$S$ (or none at all)---implying that $a$ cannot be in~$\pi(c^\star)$ under~$\pi$---we add $a$ to~$S$ and place them into one of the coalitions in~$\wt{\pi} \setminus \wt{\pi}(c^\star)$ as follows:
    \begin{itemize}
        \item If $a$ has an enemy in~$S$ that is a lonely agent in~$C$, then we place $a$ into the unique coalition in~$\wt{\pi} \setminus \wt{\pi}(c^\star)$ that does not contain this lonely agent of~$C$ (such a coalition must exist, if our guesses are correct).
        \item Otherwise, we place $a$ into some coalition in~$\wt{\pi} \setminus \wt{\pi}(c^\star)$ containing, if possible, some friend of~$a$. If $a$ has no friends, then we place them into some coalition in which $a$ has no enemies (such a coalition must exist, if our guesses are correct).
    \end{itemize}
    
    After the end of this iterative process, all remaining unplaced agents, i.e., each agent in~$\agents \setminus (\wt{\pi}(c^\star) \cup S)$ 
    either has no enemies in~$\wt{\pi}(c^\star)$ or has a friend in~$\agents \setminus (\wt{\pi}(c^\star) \cup S)$. We place all of these agents 
    into~$\wt{\pi}(c^\star)$.
    Let $\hat{\pi}$ denote the coalition structure obtained. 

    \proofsubparagraph{Correctness}
    We claim that $\hat\pi$ is IR. 
    First, notice that we never place an enemy of some lonely agent~$c_\ell \in C$ into the same coalition as~$c_\ell$: if $c_\ell \in \pi(c^\star)$, then this is clear because we place all of their enemies into the other (unique) coalition before the iterations; if $c_\ell \in S$, then this is clear due to our placement strategy during the iterative process (here we rely also on the correctness of our guesses). Moreover, our initial construction for~$\wt{\pi}$ ensures that all agents of~$\wt{C}$ that are not lonely agents in~$C$ have a friend in their coalition. It remains to see that each agent~$a$ in~$\agents \setminus \wt{C}$ either has a friend in~$\hat{\pi}(a)$ or has no enemy in~$\hat{\pi}(a)$. 
    First, if we placed $a$ during the iterative process, then either they have no friend but also no enemies in their coalition under~$\hat{\pi}$, or they have at least one friend in $\hat{\pi}(a)$.
    Second, if we placed $a$ after the iterative process into~$\wt{\pi}(c^\star)$, then by the stopping condition of the iteration, either $a$ has no enemies in~$\wt{\pi}(c^\star)$ and, hence, in~$\hat{\pi}(a)$, or they have at least one friend in $\agents \setminus S  \subseteq \hat{\pi}(a)$. This proves that $\hat{\pi}$ is IR.

    \proofsubparagraph{Running time} Finding a minimum vertex cover~$C$ for~$G^\enemies$ can be done in linear time, due to $|C| \leq 3$. There are $\Oh{1}$ possible guesses for the partitioning of~$C$ by~$\pi$ and for the lonely agents in~$C$. Guessing a friend for each agent in~$C$ that is not lonely yields at most~$\Oh{|\agents|^3}$ possibilities. %
    Our algorithm runs in~$\Oh{|\agents|^2}$ time for each set of guesses, yielding a polynomial running time. 
}

The next three intractability results show that the algorithm for \IRkfo\ in \Cref{thm:FO:k:symmetric:enmityVcAtMostThree:poly}  cannot be generalized to 
(a) deal with asymmetric preferences, even if $\vc(G^\enemies)=2$ and we ask for $k=2$ coalitions,
(b) solve \IRkfo\ for symmetric instances with $\vc(G^\enemies)=4$, or 
(c) solve \IRBSCfo\ under the same restrictions.
These results are obtained by reductions similar to the one of \Cref{thm:FO:k:twoCoals:NPh}. 

\begin{restatable}[\linkproof{thm:FO:k:enmityVcAtLeastTwo:NPh:collapsed}]{theorem}{thmFOkenmityVcAtLeastTwoNPhcollapsed}
\label{thm:FO:k:enmityVcAtLeastTwo:NPh:collapsed}%
\label{thm:FO:k:enmityVcAtLeastFour:NPh:collapsed}%
\label{thm:FO:BSC:symmetric:enmityVcAtLeastThree:NPh:collapsed}%
The followings hold: 
\begin{compactitem}
    \item[(a)]
        \IRkfo is \NPc even if $\vc(G^\enemies) = 2$, $\numCoals=2$ and $\max_{a \in \agents}|\enemies_a| \leq 1$. %
    \item[(b)]
        \IRkfo is \NPc even if $\vc(G^\enemies) = 4$ and the instance is symmetric.
    \item[(c)]
        \IRBSCfo is \NPc even on symmetric instances with $\vc(G^\enemies) = 3$. %
\end{compactitem}
\end{restatable}
\prooftoappendix{thm:FO:k:enmityVcAtLeastTwo:NPh:collapsed}{\thmFOkenmityVcAtLeastTwoNPhcollapsed*}{
We prove the result in three separate theorems.

\begin{theorem}
\label{thm:FO:k:enmityVcAtLeastTwo:NPh}
    \IRkfo is \NPc 
    even if $\vc(G^\enemies) = 2$, $\numCoals=2$ and $|\enemies_a| \leq 1$ for each $a \in \agents$.
\end{theorem}
\begin{proof}
    We prove the theorem by a modification of the reduction from \Cref{thm:FO:k:twoCoals:NPh}. We construct the same graph $G$ and remove all edges running between two clause agents. Finally, to ensure the same properties of the construction, we set $r^+\in\enemies_{c_j^+}$ and $r^-\in\enemies_{c_j^-}$ for every $j\in[m]$. Notice that the enmities are only from clause agents towards the guard agents (and between the two guard agents), and therefore every clause agent has at least one enemy in their coalition, which must be supplemented by assigning at least one of their friends.
\end{proof}

\begin{theorem}
\label{thm:FO:k:enmityVcAtLeastFour:NPh}
    \IRkfo is \NPc 
    even if $\vc(G^\enemies) = 4$ and the instance is symmetric.
\end{theorem}
\begin{proof}
    We again prove the theorem by modification of the reduction of \Cref{thm:FO:k:twoCoals:NPh}. Let $G$ be the graph resulting from this construction. We first remove all edges connecting two clause agents. Then, we add two \emph{clause guards} $g^+$ and~$g^-$. We make $g^+$ a mutual enemy with $r^-$ and~$g^-$ a mutual enemy with $r^+$. Additionally, each clause guard is a friend with all variable agents, and a mutual enemy with all clause agents of the same sign. See \Cref{fig:FO:k:enmityVcAtLeastFour:NPh:construction} for an illustration. The correctness of the construction then follows along the same lines as in the proof of \Cref{thm:FO:k:twoCoals:NPh}; the important observation is that $g^+$ and $r^+$ are necessarily in the same coalition and, therefore, each $c_j^+$ has at least one enemy in any IR coalition structure. Finally, $\{r^+,r^-,g^+,g^-\}$ clearly forms a vertex cover of $G^\enemies$ of size four.
    \begin{figure}[bt!]
        \centering
        \begin{tikzpicture}[
                agent/.style={circle, draw=black, minimum size=20pt, inner sep=1pt, font=\small,fill=white},
                friend/.style={draw=black,cbBlue, thick, dashed},
                enemy/.style={draw=black,cbRed,thick},
                every node/.style={font=\small},
                var/.style={minimum size=15pt}
            ]
            \pgfmathsetmacro\yPos{1.5}
            
            \node[agent,thick] (rp) at (0,\yPos)  {$r^+$};
            \node[agent,thick] (rn) at (0,-\yPos) {$r^-$};

            \node[agent,thick] (gn) at (10,\yPos) {$g^-$};
            \node[agent,thick] (gp) at (10,-\yPos) {$g^+$};
            
            \node[agent] (c1n) at (3,\yPos) {$c_1^-$};
            \node[agent] (c2n) at (5,\yPos) {$c_2^-$};
            \node[agent] (c3n) at (7,\yPos) {$c_3^-$};
            
            \node[agent] (c1p) at (3,-\yPos) {$c_1^+$};
            \node[agent] (c2p) at (5,-\yPos) {$c_2^+$};
            \node[agent] (c3p) at (7,-\yPos) {$c_3^+$};
            
            \node[agent, var] (a1) at (2, 0)  {$a_1$};
            \node[agent, var] (a2) at (3.5, 0) {$a_2$};
            \node[agent, var] (a3) at (5, 0)  {$a_3$};
            \node[agent, var] (a4) at (6.5, 0) {$a_4$};
            \node[agent, var] (a5) at (8, 0)  {$a_5$};

            \draw[friend] (c1p) -- (a1);
            \draw[friend] (c1p) -- (a2);
            \draw[friend] (c1p) -- (a3);
            \draw[friend] (c1n) -- (a1);
            \draw[friend] (c1n) -- (a2);
            \draw[friend] (c1n) -- (a3);
            
            \draw[friend] (c2p) -- (a1);
            \draw[friend] (c2p) -- (a3);
            \draw[friend] (c2p) -- (a4);
            \draw[friend] (c2n) -- (a1);
            \draw[friend] (c2n) -- (a3);
            \draw[friend] (c2n) -- (a4);
            
            \draw[friend] (c3p) -- (a2);
            \draw[friend] (c3p) -- (a4);
            \draw[friend] (c3p) -- (a5);
            \draw[friend] (c3n) -- (a2);
            \draw[friend] (c3n) -- (a4);
            \draw[friend] (c3n) -- (a5);

            \foreach \i in {1,2,3,4,5} {
                \draw[friend] (gn) -- (a\i.30);
                \draw[friend] (gp) -- (a\i.330);
            }

            \draw[enemy] (rp) -- (rn);
            \draw[enemy,bend left=20] (rp) to (gn); 
            \draw[enemy,bend right=20] (rn) to (gp); 
            
            \draw[enemy,bend left=22] (rp) to (c1n);
            \draw[enemy,bend left=22] (rp) to (c2n);
            \draw[enemy,bend left=22] (rp) to (c3n);
            
            \draw[enemy,bend right=22] (rn) to (c1p);
            \draw[enemy,bend right=22] (rn) to (c2p);
            \draw[enemy,bend right=22] (rn) to (c3p);
            
            \draw[enemy,bend left=22] (gp) to (c1p);
            \draw[enemy,bend left=22] (gp) to (c2p);
            \draw[enemy,bend left=22] (gp) to (c3p);
            
            \draw[enemy,bend right=22] (gn) to (c1n);
            \draw[enemy,bend right=22] (gn) to (c2n);
            \draw[enemy,bend right=22] (gn) to (c3n);

            \begin{pgfonlayer}{background}
                \filldraw[thick,dotted,cbOrange,fill=cbOrange!7,rounded corners=10pt]
                    (1, -3) -- (1,-0.5) -- (-1.25,-0.5) -- (-1.25,2.75) --
                    (2.5,2.75) -- (2.5,0.75) -- (4.5,0.75) -- (4.5,-0.5) -- 
                    (5.5,-0.5) -- (5.5,0.75) -- (7.5,0.75) -- (7.5,-0.5) --
                    (10.5,-0.5) -- (10.5,3) -- (-1.5,3) -- (-1.5,-3) -- cycle;
                \node[cbOrange] at (10,2.5) {\large$\pttn_2$};

                \filldraw[thick,dotted,cbGreen,fill=cbGreen!7,rounded corners=10pt]
                    (-1,2.5) -- (2,2.5) -- (2,0.5) -- (4,0.5) -- (4,-0.75) --
                    (6,-0.75) -- (6,0.5) -- (7,0.5) -- (7,-0.75) -- (10.5,-0.75) --
                    (10.5,-3) -- (1.5,-3) -- (1.5,0) -- (-1,0) -- cycle;
                \node[cbGreen] at (10,-2.5) {\large$\pttn_1$};
            \end{pgfonlayer}
        \end{tikzpicture}
        \caption{An illustration of the hardness construction used to prove \Cref{thm:FO:k:enmityVcAtLeastFour:NPh} for a formula $\varphi = (x_1 \lor x_2 \lor x_3) \land (x_1 \lor x_3 \lor x_4) \land (x_2 \lor x_4 \lor x_5)$. Red (solid) edges represent that two agents are enemies, while blue (dashed) edges represent friendship. Using colored backgrounds, we highlight one possible individually rational coalition structure.}
        \label{fig:FO:k:enmityVcAtLeastFour:NPh:construction}
    \end{figure}
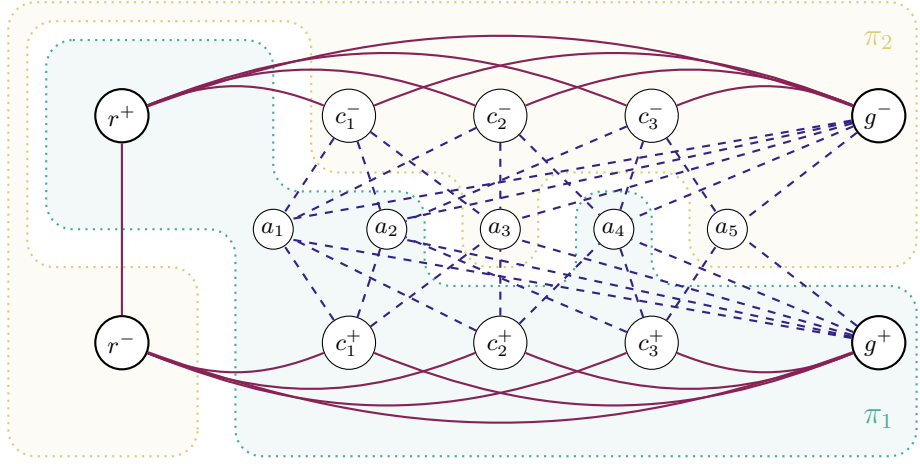
\end{proof}

\begin{theorem}
\label{thm:FO:BSC:symmetric:enmityVcAtLeastThree:NPh}
    \IRBSCfo is \NPc even if $\vc(G^\enemies) = 3$ and the instance is symmetric.
\end{theorem}
\begin{proof}
    To show \NPhness of \IRBSCfo, we present a reduction from \mbox{\probName{3-SAT}}. Let the input of this problem be a 3-CNF formula $\varphi=C_1 \wedge \dots \wedge C_m$ over a set $X=\{x_1,\dots,x_n\}$ of variables.
    Let $L=\{x_i,\overline{x}_i:i \in [n]\}$ denote the set of positive and negative literals of all variables in~$X$.

    Given $\varphi$, we construct an equivalent instance~$\mathcal{J}$ of \IRBSCfo using a construction similar to the one in the proof of \Cref{thm:FO:BSC:twoCoals:NPh}. We have a \emph{variable agent}~$a_i$ for each variable~$x_i \in X$, 
    a \emph{literal agent}~$b_\ell$  for every literal $\ell \in L$, and a \emph{clause agent}~$c_j$ for every clause~$C_j$. In addition, we create \emph{guard agents}~$r$ and~$r'$, and $n+m$ dummy agents. We let the two guard agents~$r$ and $r'$ be enemies. We also let $r'$ be enemies with all clause and variable agents. Furthermore, the agent~$a_1$ is enemies with every other variable agent and with every clause agent. Since there are no other enmities in the instance, we know that the agents $r,r',$ and~$a_1$ cover all edges in the enmity graph. 
    
    To define friendships in~$\mathcal{J}$, we let each variable agent~$a_i$ be friends with the corresponding literal agents~$b_{x_i}$ and~$b_{\overline{x}_i}$, 
    and we let each clause agent~$c_j$ be friends with the three literal agents corresponding to the three literals present in~$C_j$. 
    All remaining relations are neutral. To finalize the construction, we set $\numCoals = 2$. Note that the number of agents is~$4n+2m+2$.

    Assume first that there exists a truth assignment that satisfies~$\varphi$. Create a coalition~$\pttn_1$ that contains~$r$, all variable and clause agents, together with all literal agents corresponding to true literals. Note that~$\pttn_1$ has $2n+m+1$ agents, as desired; we put all remaining agents into another coalition~$\pttn_2$. Observe that neither~$r$ nor~$r'$ shares their coalition with an enemy. Dummies and literal agents have no enemies, so $\pi_2$ satisfies individual rationality. 
    Consider now~$\pttn_1$. Since we started from a satisfying truth assignment for~$\varphi$, each clause agent has at least one friend in their coalition. As for each variable, either its positive or its negative literal is in~$\pttn_1$, we know that every variable agent also has a friend in~$\pttn_1$. Hence, the coalition structure $(\pttn_1,\pttn_2)$ is IR.

    Assume now that there exists some IR coalition structure~$(\pttn_1,\pttn_2)$ for~$\mathcal{J}$ with $|\pttn_1|=|\pttn_2|=2n+m+1$.
    Since neither~$r$ nor~$r'$ have friends, they must be placed in different coalitions; by symmetry, we may assume $r \in \pttn_1$ and $r' \in \pttn_2$. Then all clauses and variable agents---enemies of~$r'$---must be in~$\pttn_1$. However, then all of these agents share $\pi_1$ with an enemy, due to the enmities with~$a_1$. Therefore, each of them needs at least one friend in~$\pttn_1$. Since variable agents only consider their corresponding literal agents friends, and $\pi_1$ can contain at most~$n$ literal agents due to its size, we get that $\pi_1$ contains exactly one literal agent for each variable. Moreover, since each clause agent must have at least one friend in~$\pttn_1$, the literals corresponding to the literal agents in~$\pttn_1$ yield a satisfying truth assignment for~$\varphi$.
\end{proof}
\Cref{thm:FO:k:enmityVcAtLeastTwo:NPh,thm:FO:k:enmityVcAtLeastFour:NPh,thm:FO:BSC:symmetric:enmityVcAtLeastThree:NPh} prove the \Cref{thm:FO:k:enmityVcAtLeastTwo:NPh:collapsed}.
}

We close this section with a result that shows that taking the number~$\numConflicting$ of agents that have enemies as a parameter, \IRFSCfo\ is fixed-parameter tractable.
The result is based on the observation that friendship relations between agents who have no enemies can be removed; hence, the underlying relationship graph is guaranteed to have a vertex cover of size~$\numConflicting$. As we will see in \Cref{thm:FO:FSC:vc:FPT}, such instances can be solved in \FPT time with respect to~$\numConflicting$.

\begin{restatable}[\linkproof{thm:FO:k:numConflicting:FPT}]{theorem}{thmFOknumConflictingFPT}
\label{thm:FO:k:numConflicting:FPT}
    \IRFSCfo is in \FPT when parameterized by the number of conflicting agents $\numConflicting = |\{ a \in \agents : \enemies_a \not= \emptyset \}|$. 
\end{restatable}
\prooftoappendix{thm:FO:k:numConflicting:FPT}{\thmFOknumConflictingFPT*}{
    Observe that we can safely delete all friendship relations between non-conflicting agents without changing the solvability of our instance, since every coalition structure is IR for non-conflicting agents, regardless of whether they share a coalition with a friend or not. After this modification, the set of conflicting agents is a vertex cover for the underlying graph~$G$. Hence, the result follows from \Cref{thm:FO:FSC:vc:FPT}.
}

\subsection{Restricting the Relationship Graph}
We now turn our attention to the relationship graph~$G$---including both friendships and enmities---to study how classic structural graph parameters such as treewidth and clique-width influence the computational complexity of our problems. 
To start with, we develop a fixed-parameter tractable algorithm for \IRkfo\ with parameter~$\tw(G)$, the treewidth of~$G$. Although our algorithm uses the standard method of dynamic programming along a tree-decomposition of~$G$, the details of the computation and the proof of its correctness are far from trivial.

\begin{restatable}[\linkproof{thm:FO:k:treewidth:FPT}]{theorem}{thmFOktreewidthFPT}
\label{thm:FO:k:treewidth:FPT}
    \IRkfo admits an algorithm with  $2^\Oh{\tw(G)\cdot\log\tw(G)}\cdot \numAgents$ running time.  
    Thus, \IRkfo is in \FPT parameterized by the treewidth of $G$.
\end{restatable}
\begin{proofsketch}
    For the sake of exposition, we assume that $\numCoals \leq \tw(G)$ in this proof sketch.
    We formulate an explicit bottom-up dynamic programming algorithm over a \emph{nice} tree decomposition~$\mathcal{T}$ of~$G$---a rooted tree decomposition in a standard normal form whose nodes are of four types (leaf, introduce, forget, and join nodes), simplifying the description of the computation; see the textbook of Cygan \emph{et al.}~\cite{cyg-fom-kow-lok-mar-pil-pil-sau:b:parameterized-algorithms} for a formal definition. %

    For every node $x$ of~$\mathcal{T}$, let $\beta(x)$ denote the agents in the bag of~$x$, and let $A^x$ denote the set of agents in the union of all bags in the subtree of~$\mathcal{T}$ rooted at~$x$. We compute a dynamic programming table $\DP_x[p,f,s]$, where
    \begin{itemize*}
        \item[] $p\colon \beta(x) \to [\numCoals]$ is a partition of bag agents between coalitions,
        \item[] $f\colon \beta(x) \to \{ -1, 0, 1 \}$ assigns to each bag agent a \emph{flag} based on their neighborhood, and
        \item[] $s\colon [\numCoals] \to \{0,1\}$ represents whether the coalition $j\in[k]$ is non-empty.
    \end{itemize*}
    We call the triple $(p,f,s)$ a \emph{signature}. Observe that for each node $x$, there are at most $\numCoals^\Oh{\tw}\cdot 3^\Oh{\tw} \cdot 2^\Oh{\numCoals} = 2^\Oh{\tw\cdot\log\tw}$ different signatures. For a signature $(p,f,s)$, the table $\DP_x$ stores \true if there exists a corresponding \emph{partial partition}~$\pttn'$ of $\agents^x$ such that
    \begin{enumerate}[label=\texttt{(\roman*)},left=2pt]
        \item for every agent $a\in\beta(x)$ we have $\pttn'(a) = \pttn'_{p(a)}$, that is, each bag agent is allocated to the correct coalition,\label{thm:FO:kHG:treewidth:FPT:goodCoalition:v0}
        \item for every agent $a\in\beta(x)$ such that $f(a) = -1$ holds that $|\pttn'(a)\cap \enemies_a| \geq 1$ and $|\pttn'(a)\cap\friends_a| = 0$,\label{thm:FO:kHG:treewidth:FPT:hasOnlyEnemy:v0}
        \item for every agent $a\in\beta(x)$ such that $f(a) = 0$ holds that $|\pttn'(a)\cap\enemies_a| = |\pttn'(a)\cap\friends_a| = 0$,\label{thm:FO:kHG:treewidth:FPT:hasNoRelations:v0}
        \item for every agent $a\in\beta(x)$ such that $f(a) = 1$ holds that $|\pttn'(a)\cap\friends_a| \geq 1$,\label{thm:FO:kHG:treewidth:FPT:hasFriend:v0}
        \item for every agent $a\in \agents^x\setminus\beta(x)$ the partition $\pttn'$ is IR, and\label{thm:FO:kHG:treewidth:FPT:forgottenIR:v0}
        \item for every $j\in[k]$ it holds that $|\pttn_j| \geq 1$ if $s(j)=1$ and $|\pttn_j| = 0$ otherwise.\label{thm:FO:kHG:treewidth:FPT:nonEmptyCoals:v0}
    \end{enumerate}
    If no such partial partition exists for $(p,f,s)$, the dynamic programming table stores \false for this signature. %
    
    Assume now that the dynamic programming table is computed correctly for the root node $r$. By definition, the bag of the root node is empty. Therefore, we ask if $\DP_r[\emptyset,\emptyset,\mathbf{1}]$, where $\emptyset$ represents the function with an empty domain and~$\mathbf{1}$ is a constant-$1$ function, is set to \true. If it is the case, then there exists a partition of $\agents^r = \agents$ which is individually rational for every agent $a\in\agents$ by property \ref{thm:FO:kHG:treewidth:FPT:forgottenIR:v0} and all coalitions are non-empty by property \ref{thm:FO:kHG:treewidth:FPT:nonEmptyCoals:v0}. So, the algorithm returns~\Yes. Otherwise, no such partition exists, and, therefore, the algorithm correctly returns~\No.
\end{proofsketch}

\prooftoappendix{thm:FO:k:treewidth:FPT}{\thmFOktreewidthFPT*}{
    We split the algorithm into two different cases based on the relation between $\numCoals$ and $\tw(G)$. Already Fioravantes \emph{et al.}~\cite{FioravantesGMS2026a} noticed that for an additively separable hedonic game  (a class of hedonic games that includes our friends--enemies--neutral model) with underlying relationship graph~$G$, there always exists an IR coalition structure with $\numCoals$ non-empty coalitions, if $\numCoals \geq \tw(G) + 1$ (and $\numCoals \in [\numAgents]$); moreover, such a coalition structure can found in \FPT time parameterized by $\tw(G)$. The argument is based on the relationship between our problem and coloring: whenever graph $G$ is of treewidth at most $\tau$, then it is colorable by $\tau+1$-many colors~\citeapp{Chlebikova2002}. Moreover, such a coloring can be found in time $2^\Oh{\tau \cdot\log\tau}\cdot n$~\citeapp{ArnborgP1989}. This coloring then partitions the agents into independent sets, and we can keep splitting these sets until we achieve the desired number of coalitions. 
    That is, if $\numCoals \geq \tw(G)+1$, we can solve the problem in \FPT time.
    Hence, for the rest of the proof, we assume that $\numCoals \leq \tw(G)$. 

    We formulate an explicit bottom-up dynamic programming algorithm over a nice tree decomposition~$\mathcal{T}$ of~$G$. We can assume that the decomposition is given as part of the input, as if not, we can compute one of width at most $2\cdot\tw(G)$ in~$2^\Oh{\tw}\cdot n$ time using Korhonen's algorithm~\citeapp{Korhonen2021}.

    For every node $x$ of the tree decomposition~$\mathcal{T}$, let $\beta(x)$ denote the agents in the bag of~$x$, and let $A^x$ denote the set of agents in the union of all bags in the subtree of~$\mathcal{T}$ rooted at~$x$. We compute a dynamic programming table $\DP_x[p,f,s]$, where
    \begin{itemize}
        \item $p\colon \beta(x) \to [\numCoals]$ is a partition of bag agents between coalitions,
        \item $f\colon \beta(x) \to \{ -1, 0, 1 \}$ assigns to each bag agent a \emph{flag} based on their neighborhood, and
        \item $s\colon [\numCoals] \to \{0,1\}$ represents whether the coalition $j\in[k]$ is non-empty.
    \end{itemize}
    We call the triple $(p,f,s)$ a \emph{signature}. Observe that for each node $x$, there are at most $\numCoals^\Oh{\tw}\cdot 3^\Oh{\tw} \cdot 2^\Oh{\numCoals} = 2^\Oh{\tw\cdot\log\tw}$ different signatures. For a signature $(p,f,s)$, the table $\DP_x$ stores \true if there exists a corresponding \emph{partial partition}~$\pttn'$ of $\agents^x$ such that
    \begin{enumerate}[label=\texttt{(\roman*)},left=2pt]
        \item for every agent $a\in\beta(x)$ we have $\pttn'(a) = \pttn'_{p(a)}$, that is, each bag agent is allocated to the correct coalition,\label{thm:FO:kHG:treewidth:FPT:goodCoalition}
        \item for every agent $a\in\beta(x)$ such that $f(a) = -1$ holds that $|\pttn'(a)\cap \enemies_a| \geq 1$ and $|\pttn'(a)\cap\friends_a| = 0$,\label{thm:FO:kHG:treewidth:FPT:hasOnlyEnemy}
        \item for every agent $a\in\beta(x)$ such that $f(a) = 0$ holds that $|\pttn'(a)\cap\enemies_a| = |\pttn'(a)\cap\friends_a| = 0$,\label{thm:FO:kHG:treewidth:FPT:hasNoRelations}
        \item for every agent $a\in\beta(x)$ such that $f(a) = 1$ holds that $|\pttn'(a)\cap\friends_a| \geq 1$,\label{thm:FO:kHG:treewidth:FPT:hasFriend}
        \item for every agent $a\in \agents^x\setminus\beta(x)$ the partition $\pttn'$ is IR, and\label{thm:FO:kHG:treewidth:FPT:forgottenIR}
        \item for every $j\in[k]$ it holds that $|\pttn_j| \geq 1$ if $s(j)=1$ and $|\pttn_j| = 0$ otherwise.\label{thm:FO:kHG:treewidth:FPT:nonEmptyCoals}
    \end{enumerate}
    If no such partial partition exists for $(p,f,s)$, the dynamic programming table stores \false for this signature. We use $\pttn^p$ to denote the partition of $\beta(x)$ induced by $p$.

    Assume now that the dynamic programming table is computed correctly for the root node $r$. By definition, the bag of the root node is empty. Therefore, we ask if $\DP_r[\emptyset,\emptyset,\mathbf{1}]$, where $\emptyset$ represents the function with an empty domain and~$\mathbf{1}$ is a constant-$1$ function, is set to \true. If it is the case, then there exists a partition of $\agents^r = \agents$ which is individually rational for every agent $a\in\agents$ by property \ref{thm:FO:kHG:treewidth:FPT:forgottenIR} and all coalitions are non-empty by property \ref{thm:FO:kHG:treewidth:FPT:nonEmptyCoals}. So, the algorithm returns~\Yes. Otherwise, no such partition exists, and, therefore, the algorithm correctly returns~\No.

    We now formally define the computation of the dynamic programming table separately for each type of node of the tree decomposition.

    \proofcase[Leaf]{Node}{}
    Let $x$ be a leaf node, which means that  $\beta(x) = \emptyset$. Hence, we directly determine the value of the table for each signature $(p,f,s)$, with $p$ and~$f$, necessarily having an empty domain, 
    solely on the value of $s$ by setting
    \[
        \DP_x[p,f,s] = \begin{cases}
             \true & \text{if } \forall j\in[\numCoals] \colon s(j) = 0 \text{ and}\\
            \false & \text{otherwise.}
        \end{cases}
    \]
    That is, the stored value is \true if and only if $s$ prescribes all coalitions to be empty.

    \proofcase[Introduce]{Node}{}
    Let $x$ be an introduce node with a single child $y$. By definition, we have $\beta(x) = \beta(y) \cup \{a\}$ for some $a \not\in \beta(y)$. Based on the part the newly introduced agent is placed to, we make sure that the flags of all other members of this part are valid and that $a$'s coalition is not marked as empty. If this is the case, we ask if there is a solution for the same partition where $a$ is removed from their part. Formally, the computation is as follows:
    \begin{gather}
    \raisetag{16pt}
        \DP_x[p,f,s] = \left\{
        \begin{array}{ll}
            \multicolumn{2}{l}{
            \false \qquad  \text{if } s(p(a)) = 0,}
            \\[4pt]
            \multicolumn{2}{l}{
            \false \qquad \text{if } f(a) = -1 \land ((\pttn^p(a) \cap \enemies_a = \emptyset) 
            \lor 
            (\pttn^p(a) \cap \friends_a \neq \emptyset))} 
            \\[4pt]
            \multicolumn{2}{l}{
            \false \qquad \text{if } f(a) = 0 \land ((\pttn^p(a) \cap \enemies_a \neq \emptyset) 
            \lor 
            (\pttn^p(a) \cap \friends_a \neq \emptyset))} 
            \\[4pt]
            \multicolumn{2}{l}{
            \false \qquad \text{if } f(a) = 1 \land (\pttn^p(a) \cap \friends_a = \emptyset)} 
            \\[4pt]
            \multicolumn{2}{l}{
            \false \qquad \text{if } \exists b \in \pttn^p(a) \colon f(b)\neq 1 \land (a \in \friends_b)} 
            \\[6pt]
            \multirowcell{2}{
            \bigvee_{\substack{
                f^y\colon \beta(y)\to\{-1,0,1\}\\
                \forall b\in\beta(y) \setminus(\pttn^p(a) \cap (\friends_a\cup\enemies_a))\colon\\
                f^y(b) = f(b)\\
                \forall b \in \beta(y)\colon \  (f^y(b)=1) \Rightarrow (f(y)=1)\\
                \forall b\in \pttn^p(a), a \in \enemies_b \colon\\
                (f^y(b)\in \{-1,0\}) \Rightarrow (f(b)=-1) 
            }}}
            & 
            \DP_y[ p_{\downarrow \beta(y)}, f^y, s \Join (p(a) \mapsto 0 ) ]  
            \\[2pt]
            & \phantom{xx}\lor\DP_y[ p_{\downarrow \beta(y)}, f^y, s ] \qquad  \text{otherwise.}
            \\[36pt]
        \end{array}
        \right.
        \label{eq:introduce_node_comp}
    \end{gather}
    In the definition of the computation, we use $g_{\downarrow X}$ to denote the restriction of a function~$g$ to the domain~$X$. Moreover, by $g \Join (x \mapsto y)$ we denote the function that coincides with~$g$ on all elements of the domain of~$g$ other than~$x$, and maps~$x$ to~$y$.

    \proofcase[Forget]{Node}{}
    Let $x$ be a forget node with a single child $y$ and let $a\in\beta(y)$ be such that $\beta(x) = \beta(y)\setminus\{a\}$. Here, we check that there is a solution which $a$ considers to be individually rational. This can be read from the dynamic programming table of the child $y$ as follows.
    \begin{equation}
        \label{eq:forget_node_comp}
        \DP_x[p,f,s] = \bigvee_{j\in[\numCoals]} \bigvee_{b\in\{0,1\}} \DP_y[ p \Join ( a \mapsto j ), f \Join ( a \mapsto b ), s ]\,.
    \end{equation}
    The $\Join$ operator is the same as in the case of introduce node, just notice that this time it extends the domain of the function.

    \proofcase[Join]{Node}{}
    A join node $x$ has exactly two children $y$ and $z$ such that $\beta(x) = \beta(y) = \beta(z)$. In the join node, we merge the solutions with the same partition of bag vertices from $y$ and $z$. For this, we need to carefully handle the flags.
    \begin{gather}
        \label{eq:join_node_comp}
        \raisetag{16pt}
        \DP_x[p,f,s] = \bigvee_{\substack{
                f^y,f^z\colon\beta(x)\to\{-1,0,1\}\\
                \forall a\in\beta(x) \text{ s.t. } f(a) = -1\colon f^y(a) + f^z(a) \leq -1\\
                \forall a\in\beta(x) \text{ s.t. } f(a) = 0\colon f^y(a) = f^z(a) = 0\\
                \forall a\in\beta(x) \text{ s.t. } f(a) = 1\colon \max\{f^y(a),f^z(a)\} = 1\\
                s^y,s^z\colon[\numCoals]\to\{0,1\}\\
                \forall j\in[k]\colon\max\{s^y(j),s^z(j)\} = s(j)
            }} 
                \DP_y[ p, f^y, s^y ] \land \DP_z[ p, f^z, s^z ]
    \end{gather}

    \proofsubparagraph{Correctness}
    Now, we formally prove the correctness of the computation of the dynamic programming algorithm. We proceed by mathematical induction over the type of the node of the tree decomposition, and we split the proof into two claims.

    \begin{lemma}
        Let $x$ be a node. If, for some signature $(p,f,s)$, we have $\DP_x[p,f,s] = \true$, then there exists a partial solution $\pttn^x$ corresponding to the signature $(p,f,s)$.
    \end{lemma}
    \begin{claimproof} 
        We use bottom-up induction on~$\mathcal{T}$; hence, we assume that the claim holds for the children of~$x$.

        First, assume that $x$ is a leaf node. Since $\beta(x) = \emptyset$, the functions $p$ and~$f$ have an empty domain. Moreover, if $\DP_x[p,f,s] = \true$, then by the definition of the computation, we have $s = \vec{0}$. If we take $\pttn^x$ so that $\pttn^x_j = \emptyset$ for every $j\in[\numCoals]$, we have a partial partition of $\agents^x$ corresponding to $(p,f,s)$: all coalitions are empty as prescribed by $s$, and all other properties are satisfied trivially.

        Consider the case when $x$ is an introduce node with $y$ as its unique child and $\beta(x)=\beta(y) \cup \{a\}$. 
        Since $\DP_x[p,f,s]=\true$, there exists some flag function~$f^y$ satisfying the conditions of the computation in~(\ref{eq:introduce_node_comp}) and some $s^y$ with $s^y=s$ or $s^y=s \Join (p(a) \mapsto 0)$ such that  $\DP_y[p_{\downarrow\beta(y)},f^y,s^y]=\true$. Let $\pttn^y$ be a partial partition corresponding to $\sigma=(p_{\downarrow\beta(y)},f^y,s^y)$ at~$y$, and let $\pttn$ be the partitioning of~$A^x$ that coincides with~$\pttn^y$ on~$A^y$ and puts $a$ into coalition~$\pttn_{p(a)}$. 
        Then property~\ref{thm:FO:kHG:treewidth:FPT:goodCoalition} clearly holds. Properties~\ref{thm:FO:kHG:treewidth:FPT:hasOnlyEnemy}--\ref{thm:FO:kHG:treewidth:FPT:hasFriend} are satisfied for~$a$ as otherwise the computation would have set $\DP[p,f,s]=\false$. It is not hard to verify that $\pttn^y$ corresponds to~$\sigma$, these properties also hold for all other agents in~$\beta(x)$, due to our conditions on~$f^y$ and on~$f$. 
        Property~\ref{thm:FO:kHG:treewidth:FPT:forgottenIR} remains true for~$\pttn$ trivially. 
        Property~\ref{thm:FO:kHG:treewidth:FPT:nonEmptyCoals} also holds: each coalition other than the~$p(a)$-th is empty in~$\pttn$ if and only if it is empty in~$\pttn^y$, and accordingly, $s$ and~$s^y$ coincide on these values; in addition, the $p(a)$-th coalition is not empty, as it contains~$a$, and indeed, $s(p(a))=1$ is ensured. 
    
        Now, let $x$ be a forget node with a unique child~$y$ and $\beta(x)=\beta(y) \setminus \{a\}$. 
        Since $\DP_x[p,f,s]=\true$, by the computations in~$(\ref{eq:forget_node_comp})$ there exists a coalition $j\in[\numCoals]$ and some $b\in\{0,1\}$ such that $\DP_y[p\Join(a\mapsto j),f\Join(a\mapsto b),s]=\true$. Let $\pttn^y$ be a partial partition corresponding to~$\sigma=(p\Join(a\mapsto j),f\Join(a\mapsto b),s)$ at~$y$. We claim that $\pttn^y$ also corresponds to the signature $(p,f,s)$ at~$x$. We have $\agents^x = \agents^y$, so $\pttn^y$ is  a partition of $\agents^x$.
        Properties~\ref{thm:FO:kHG:treewidth:FPT:goodCoalition}--\ref{thm:FO:kHG:treewidth:FPT:hasFriend}  and~\ref{thm:FO:kHG:treewidth:FPT:nonEmptyCoals} clearly remain true for~$\pttn^y$ with respect to~$\sigma$ as well.
        Property~\ref{thm:FO:kHG:treewidth:FPT:forgottenIR} remains true for~$\pttn^y$ w.r.t. $\sigma$, because $b \in \{0,1\}$ and thus $\pttn^y$ is IR for~$a$, proving the claim.
        
        Finally, let $x$ be a join node with children~$y$ and~$z$. Since $\DP_x[p,f,s]=\true$, there exist signatures $\sigma_y=(p,f^y,s^y)$ and $\sigma_z=(p,f^z,s^z)$
        such that $\DP_y[p,f^y,s^y]=\true$ and
        $\DP_z[p,f^z,s^z]=\true$,
        with the functions $f^y,f^z,s^y$, and~$s^z$ satisfying the conditions of our computation in~(\ref{eq:join_node_comp}). Let $\pttn^y$ and $\pttn^z$ be partial partitions corresponding to $\sigma_y$ at~$y$ and to $\sigma_z$ at~$z$, respectively. Create the partial partition~$\pttn=(\pttn^y_1 \cup \pttn^z_1,\dots,\pttn^y_k \cup \pttn^z_k)$ for~$A^x=A^y \cup A^z$. Note that property~\ref{thm:FO:kHG:treewidth:FPT:goodCoalition}  trivially holds for~$\pttn$, and properties~\ref{thm:FO:kHG:treewidth:FPT:hasFriend}--\ref{thm:FO:kHG:treewidth:FPT:hasOnlyEnemy} also remain true: for agents in~$A^x \setminus \beta(x)$ this holds because $\mathcal{T}$ is a tree-decomposition, and thus all agents in~$A^y \setminus \beta(x)$ is neutral toward all agents in~$A^z \setminus \beta(x)$;
        for agents in~$\beta(x)$ this holds because of our conditions on~$f^z, f^z,$ and~$f$. Property~\ref{thm:FO:kHG:treewidth:FPT:forgottenIR} holds up trivially, and property~\ref{thm:FO:kHG:treewidth:FPT:nonEmptyCoals} holds because a coalition in~$\pttn_j=\emptyset$ for some~$j \in [k]$ if and only if $\pttn^y_j=\pttn^z_j=\emptyset$. This finishes the proof of the claim for join nodes. 

        Hence, by induction, the lemma holds.
     \end{claimproof}

    \begin{lemma}
        Let $x$ be a node. If there is a partial solution $\pttn^x$ corresponding to some signature $(p,f,s)$, then $\DP_x[p,f,s] = \true$.
    \end{lemma}
    \begin{claimproof}
        We use bottom-up induction on~$\mathcal{T}$; hence, we assume that the claim holds for the children of~$x$.
        Let $\pttn^x$ be a partial partition corresponding to  $\sigma=(p,f,s)$. 

        First, let $x$ be a leaf node. Since $\beta(x) = \emptyset$ in this case, the only possible partial solution for $\agents^x$ is the empty partition. 
        That is, necessarily $s(j) = 0$ for every $j\in[\numCoals]$. However, in this case, the dynamic programming table stores \true for $\sigma$. That is, for leaf nodes, the claim holds.

        Consider now an introduce node~$x$ with unique child~$y$ and $\beta(x)=\beta(y) \cup \{a\}$. 
        Let $j \in [k]$ be the coalition for which $\pttn^x(a)=\pttn^x_j$, and let $\pttn^y$ be the partial solution for~$A^y=A^x \setminus \{a\}$ obtained from~$\pttn^x$ by removing~$a$ from~$\pttn^x_j$. If $\pttn^y_j=\emptyset$, then we set~$s^y=s \Join (p(a) \mapsto 0)$, otherwise we set~$s^y=s$. 
        In addition, we define a flag function~$f^y$ as follows: for each $b \in \beta(y)$, we set $f^y(b)=f(b)$ 
        unless $b \in \pttn^x_j$ and either~$b$ has no friends or enemies other than~$a$ in~$\pttn^x_j$, or $b$ has $a$ as their unique friend in~$\pttn^x_j$ but $\enemies_b \cap \pttn^x_j \neq \emptyset$.
        In the former case, we set~$f^y(b)=0$, and in the latter case we set $f^y(b)=-1$.
        It is then straightforward to verify that $\pttn^y$ is a partial partition for~$A^y$ that corresponds to the signature $\sigma'=(p_{\downarrow \beta(y)},f^y,s^y)$. Thus, by induction, we know that $\DP_y[\sigma']$ is set to \true by the algorithm, and hence is considered in~(\ref{eq:introduce_node_comp}) when computing $\DP_x[\sigma]$. Since $\pttn^x$ is a partial partition for~$A^x$ corresponding to~$(p,f,s)$, by properties~\ref{thm:FO:kHG:treewidth:FPT:goodCoalition}--\ref{thm:FO:kHG:treewidth:FPT:hasFriend} and~\ref{thm:FO:kHG:treewidth:FPT:nonEmptyCoals} we know that none of the conditions for setting $\DP_x[\sigma]$ to \false holds. Hence, when considering the signature~$\sigma'$, the algorithm finds that the conditions on~$f^y$ and~$f$ in~(\ref{eq:introduce_node_comp}) hold. Thus, $\DP_x[\sigma]$ is set to \true.

        Next, let $x$ be a forget node with unique child~$y$ and $\beta(x)=\beta(y) \setminus \{a\}$. Let $j \in [k]$ be the coalition such that $\pttn^x(a) = \pttn^x_j$.
        Set $p^y=p\Join (a\mapsto j)$ and $f^y=f\Join(a\mapsto b)$  where $b = 1$ if $\friends_a\cap\pttn^x_j \neq \emptyset$ and $0$ otherwise.
        We now argue that $\pttn^x$ is also a partial partition corresponding to $\sigma^y=(p^y,f^y,s)$ in $y$.
        Property~\ref{thm:FO:kHG:treewidth:FPT:goodCoalition} for~$\sigma^y$ is clearly satisfied by our choice of~$j$. 
        Properties~\ref{thm:FO:kHG:treewidth:FPT:hasOnlyEnemy}--\ref{thm:FO:kHG:treewidth:FPT:hasFriend} remain true for all agents in~$\beta(x)$; it suffices to check it for agent~$a$. Notice that since property~\ref{thm:FO:kHG:treewidth:FPT:forgottenIR} holds for~$\pttn^x$ with respect to~$x$, we know that  
        $\pttn^x$ is IR for~$a$. Thus, either $a$ has a friend in~$\pttn^x_j$, or they have no enemy in~$\pttn^x_j$. Observe that $f^y(a) = 1$ in the first case  and 
        $f^y(a)=0$ in the second case by properties~\ref{thm:FO:kHG:treewidth:FPT:hasFriend} and~\ref{thm:FO:kHG:treewidth:FPT:hasNoRelations} for~$\sigma$, respectively. Hence, $\pttn^y$ satisfies properties~\ref{thm:FO:kHG:treewidth:FPT:hasOnlyEnemy}--\ref{thm:FO:kHG:treewidth:FPT:hasFriend} for~$\sigma^y$. Properties~\ref{thm:FO:kHG:treewidth:FPT:forgottenIR} and~\ref{thm:FO:kHG:treewidth:FPT:nonEmptyCoals} hold for~$\sigma^y$ as well. 
        Thus, by the induction hypothesis, $\DP_y[p^y,f^y,s]$ is set to \true. When computing $\DP_x$ in~(\ref{eq:forget_node_comp}), the algorithm considers $\DP_y[p^y,f^y,s]$, thus $\DP_x[p,f,s]$ is set to \true, as required. %

        Finally, let $x$ be a join node with children~$y$ and~$z$ and $\beta(x)=\beta(y)=\beta(z)$. For both $u \in \{y,z\}$, define the partial partition~$\pttn^u$ by setting $\pttn^u_j=\pttn^x_j \cap A^u$ for each $j \in [k]$.
        Moreover, construct the flag function $f^u$ by setting~$f^u(b)=1$ if $\friends_b \cap \pttn^u(b) \neq \emptyset$, setting $f^u(b)=0$ if $\friends_b\cap \pttn^u(b)=\enemies_b\cap \pttn^u(b)=\emptyset$, and setting $f^u(b)=-1$ otherwise.
        Let us also $s^u(j)=1$ if $\pttn^u_j \neq \emptyset$ and $s^u(j)=0$ otherwise. Then it is clear that $\pttn^u$ is a partial solution for $\sigma^u=(p,f^u,s^u)$ at~$u$.
        Hence, by induction, $\DP_u[(p,f^u,s^u]$ is set to \true for both $u \in \{y,z\}$. Thus, when computing $\DP_x[\sigma]$ in~(\ref{eq:join_node_comp}), the algorithm considers the signatures $\sigma^y$ and~$\sigma^z$, and moreover, it is not hard to see that the functions~$f^y,f^z,s^y$, and~$s^x$ satisfy the conditions in~(\ref{eq:join_node_comp}). Therefore, the algorithm sets $\DP_x[\sigma]$ to \true, finishing the claim for a join node.

        By induction, the lemma holds.
    \end{claimproof}

    \proofsubparagraph{Running time}
    For the running time, it is easy to see that the computation of a single cell in leaf nodes takes $\Oh{\numCoals}$ time. The same running time is required to determine the value of a single cell in forget nodes. More time-consuming are introduce and join nodes. In the former, we try all valid flag functions, and there are $3^{\tw(G)} \in 2^\Oh{\tw(G)}$ many of them. In the latter, we again try all possible flag functions, and additionally, we try all possible size functions. That is, the running time in join nodes is $2\cdot 3^{\tw(G)} \cdot 2\cdot2^{\tw(G)} = 2^\Oh{\tw(G)}$. Overall, there are $n\cdot (\numCoals^{\tw(G)} \cdot 3^{\tw(G)} \cdot 2^{\tw(G)})\in n\cdot \numCoals^\Oh{\tw(G)}$ different cells. Thus, since $\numCoals \leq \tw(G)$, the algorithm runs in $2^\Oh{\tw(G)\cdot\log\tw(G)}\cdot n$ time, which is clearly in \FPT.
}

We can extend our dynamic programming algorithm on the tree-decomposition of the relationship graph~$G$ so that for each coalition that has an agent in the current bag, we store the number of agents contained in a partial solution; this approach yields an \XP algorithm for our size-constrained problem variants with parameter~$k+\tw(G)$. 
\begin{restatable}[\linkproof{thm:FO:SCC:treewidth:XP}]{theorem}{thmFOSCCtreewidthXP}
\label{thm:FO:SCC:treewidth:XP}
    \IRSCCfo is in \XP parameterized by~$k+\tw(G)$.
\end{restatable}
\prooftoappendix{thm:FO:SCC:treewidth:XP}{\thmFOSCCtreewidthXP*}{
Let us show how we can extend the algorithm presented in \Cref{thm:FO:k:treewidth:FPT} to solve \IRSCCfo.
First, we need to adjust the notion of a signature: we let the function~$s$ take values from~$[\numAgents]$ for each coalition in~$[k]$, representing the size of the coalition in a corresponding partial partition.
More precisely, a partial partition corresponding signature~$(p,f,s)$ must satisfy conditions \ref{thm:FO:kHG:treewidth:FPT:goodCoalition}--\ref{thm:FO:kHG:treewidth:FPT:forgottenIR} and, instead of~\ref{thm:FO:kHG:treewidth:FPT:nonEmptyCoals}, fulfills the following: 
\begin{description}
    \item[\quad \texttt{(vi')}] for every $j \in [k]$ it holds that $s(j)=|\pttn'_j|$.
\end{description}

Next, we give the necessary modifications for the computation of the table~$\DP_x$ for some node~$x$.
The case when $x$ is a leaf node needs no modifications, since every partial partition is empty.

When $x$ is an introduce node, we set $\DP_x(p,f,s)$ to be \false in all cases as described in~(\ref{eq:introduce_node_comp}), and in the remaining case we set $\DP_x(p,f,s)=\bigvee \DP_y[p_{\downarrow \beta(y)},f^y,s']$ where the conditions on~$f^y$ are exactly as in~(\ref{eq:introduce_node_comp}) and $s'(j)$ equals $s(j)$ for all values of~$[k] \setminus \{p(a)\}$, and $s'(p(a))=s(p(a))-1$, capturing the condition that introducing agent~$a$ and putting them into coalition~$p(a)$ increases the size of this coalition by exactly one.

The case when $x$ is a forget node is unchanged. 

The case when $x$ is a join node needs the slight modification that the functions~$s^y$ and~$s^z$ must satisfy the condition that $s(j)=s^y(j)+s^z(j)-|\{b \in \beta(x)\colon p(b)=j\}|$.

Once $\DP_r$ is computed for the root node~$r$, it suffices to check whether $\DP_r(\emptyset,\emptyset,s)=\true$ for some size function~$s$ for which $\lb(j) \leq s(j) \leq \ub(j)$ for each $j \in [k]$.  
It is straightforward to verify the correctness of the thus modified algorithm. To calculate the running time, note that the computation in each node is dominated by the case for a join node, where we need to try all possible flag functions, which means $2^{\Oh{\tw(G)}}$ possibilities, and all possible functions for~$s^y$ and~$s^z$ meeting the conditions. Note that given $p$, we can compute~$s^z$ if we fix~$s^y$; hence, there are $\numAgents^{k}$ possibilities. 
The number of cells to compute is upper-bounded by 
$\numAgents \cdot ( k^{\tw(G)} \cdot 3^{\tw(G)} \cdot \numAgents^{k})$, which yields an overall running time of $n^{\Oh{k}} \cdot k^{\Oh{\tw(G)}}$, which is \XP with respect to parameter~$k+\tw(G)$. In fact, the running time is \FPT with parameter~$\tw(G)$ for each constant value of~$k$.
}

Based on our reductions from the \probName{$q$-Coloring} problem---the generalization of \probName{$3$-Coloring} in which the vertices must be properly colored using $q$~colors---we can show that our algorithm presented in~\Cref{thm:FO:k:treewidth:FPT} has optimal running time, assuming the Exponential Time Hypothesis (ETH)~\cite{imp-pat-zan:j:ETH}. In addition, the same reduction allows us to observe that \IRkfo\ is \Wh[1] when parameterized by the clique-width of~$G$. 
These intractability results
also extend to the \IRBSCfo\ problem, due to known results on the \probName{Equitable $q$-Coloring} problem.

\begin{restatable}[\linkproof{thm:FO:k:treewidth:tight:collapsed}]{theorem}{thmFOktreewidthtightcollapsed}
    \label{thm:FO:k:treewidth:tight:collapsed}%
    \label{thm:FO:k:cliquewidth:Wh:collapsed}%
    \label{thm:FO:BSC:treedepth:Wh:collapsed}
    The followings hold:
    \begin{compactitem}
        \item[(a)] 
        Unless ETH fails, no algorithm solves FO-$\numCoals$-HG in $2^{o(\tw(G)\cdot\log\tw(G))}\cdot n$ time.
        \item[(b)]
        \IRkfo  is \Wh when parameterized by the clique-width of~$G$.
        \item[(c)] 
        \IRBSCfo is \Wh when parameterized by the treewidth of $G$ and the number~$\numCoals$ of coalitions, combined, and \pNPh when parameterized by the clique-width of $G$.
    \end{compactitem}
\end{restatable}

\prooftoappendix{thm:FO:k:treewidth:tight:collapsed}{\thmFOktreewidthtightcollapsed*}{
    Let us first prove~(a).
    It is known that \probName{$q$-Coloring} cannot be solved in $2^{o(\tw(H)\cdot\log\tw(H))}\cdot |V(H)|$ time unless ETH fails~\citeapp{LokshtanovMS2018}. For the sake of contradiction, assume that there is an algorithm $\mathcal{A}$ running in $2^{o(\tw(G)\log\tw(G))}\cdot n$ time for \IRkfo. %
    Given an instance $\mathcal{I}=(H,q)$ of \probName{$q$-Coloring}, %
    we construct an equivalent instance~$\mathcal{J}$ of \IRkfo  as in \Cref{thm:FO:k:fourEnemies:NPh} (just starting with \probName{$q$-Coloring} instead of \probName{$3$-Coloring}). Observe that the reduction can be done in linear time and, moreover, the treewidth of~$G$ equals~$\tw(H)$, as the reduction does not change the underlying graph. Then, we use $\mathcal{A}$ to solve the instance $\mathcal{J}$, and return the same response for $\mathcal{I}$. Since $\mathcal{I}$ is equivalent to $\mathcal{J}$, the response is correct. Moreover, the overall running time of such an algorithm for \probName{$q$-Coloring} is $O(|V(H)|)+2^{o(\tw(G)\log\tw(G))}\cdot \numAgents = 2^{o(\tw(H)\log\tw(H))}\cdot |V(H)|$, which contradicts ETH. Hence, such an algorithm $\mathcal{A}$ cannot exist, which finishes the proof.

    Claim~(b) again follows from the reduction from \probName{$q$-Coloring}, presented in \Cref{thm:FO:k:fourEnemies:NPh}. The reduction does not change the underlying graph $G$ and therefore, since \probName{$q$-Coloring} is \Wh with respect to the clique-width~\citeapp{FominGLS2010},  \IRkfo is also \Wh when parameterized by this parameter.

    For claim~(c), the argument is the same as for~(b), we just start with the \probName{Equitable $q$-Coloring} problem, which is \Wh when parameterized by the treewidth plus the number of colors~\citeapp{FellowsFLRSST2011} and \pNPh when parameterized by the clique-width~\citeapp{GomesGS2023}.
}

Motivated by the intractability of \IRBSCfo\ with respect to treewidth, as shown in \Cref{thm:FO:BSC:treedepth:Wh:collapsed}, we consider the vertex cover number of the relationship graph as a stronger parameter, leading to an \FPT algorithm for \IRFSCfo.

\begin{theorem}\label{thm:FO:FSC:vc:FPT}
    When parameterized by $\vc(G)$, \IRFSCfo is in \FPT.%
\end{theorem}
\begin{proof}
    Let $M$ denote a vertex cover of~$G$ of minimum size; we can compute~$M$ in $1.25284^{|M|}\cdot(\numAgents+|E|)$ time~\cite{HarrisN2024}. We will assume that there is an IR coalition structure~$\pi$ for our instance.
    We call an agent~$a \in M$  \emph{lonely} in~$\pi$  if $\friends_a \cap \pi(a)=\emptyset$. 
    For each agent~$a \in M$ that is not lonely, we fix some agent~$a_c \in \friends_a \cap \pi(a)$. Let~$A'=\{a_c:c \in M, c$ is not lonely in~$\pi\}$.
    
    Our algorithm makes four guesses in total: the partitioning of~$M$ induced by~$\pi$, the set of lonely agents, an auxiliary graph~$H$, and, later on, an ordering of the parts of the guessed partitioning; the number of possible combinations is bounded in the running-time analysis at the end of the proof.
    We start by guessing the partitioning~$\wt{\pi}$ of~$M$ induced by~$\pi$. Next, for each agent~$a$ in~$M$, we guess whether $a$ is lonely in~$\pi$; let $M'$ denote the agents in~$M$ that are \emph{not} lonely. Then, we guess the following auxiliary graph~$H$ whose vertex set is~$M' \cup A'$ and whose edge set is $\{(c,a_c):c \in M'\}$. 
    For each coalition~$\wt{C} \in \wt{\pi}$, we let $S(\wt{C})=\wt{C} \cup \{a_c:c \in M' \cap \wt{C}\}$. Assuming that our guesses are correct,  we have $S(\wt{C}) \subseteq \pi(\wt{C})$ where $\pi(\wt{C})$ denotes the unique coalition in~$\pi$ that contains all agents in~$\wt{C}$. Although we do not know the identity of the agents in~$S(\wt{C})$, we know the cardinality $|S(\wt{C})|$.

\def\ind{\mathsf{ind}}    
    
    Next, we decide which coalitions we should allocate the agents of~$S(\wt{C})$ to. 
    We may assume that %
    $\ub_i \leq \ub_j$ whenever $1 \leq i \leq j \leq k$.
    Let us now guess the ordering $\wt{C}_1,\dots,\wt{C}_\ell$ of the 
    sets in~$\wt{\pi}$ according to the size of the coalitions they are contained in under~$\pi$, so that 
    $|\pi(\wt{C}_i)| \leq |\pi(\wt{C}_j)|$
    whenever $1 \leq i< j \leq \ell$. 
    Now, we compute a mapping~$\ind:[\ell] \rightarrow [k]$ as follows.
    For $i=1,\dots,\ell$, we  find a coalition index~$\ind(i) \in [k]$ such that
    \begin{enumerate*}[label=(\roman*)]
    \item $\ind(i) > \ind(i-1)$ for every $i>1$,
    \item $\ub_{\ind(i)} \geq |S(\wt{C}_i)|$, and
    \item $\ub_{\ind(i)}$ is as small as possible.%
    \end{enumerate*}

    \begin{restatable}[\linkproof{clm:locate_coalition_index}]{claim}{clmlocatecoalitionindex}
    \label{clm:locate_coalition_index}
        Suppose that an IR coalition structure~$\pi$ exists, and the algorithm made correct guesses with respect to~$\pi$. 
        Then there exists an IR coalition structure~$\pi'$ where, for each $i \in [\ell]$, the coalition $\pi(\wt{C}_i)$ has size~$\ub_{\,\ind(i)}$.
    \end{restatable}
    \prooftoappendix{clm:locate_coalition_index}{\clmlocatecoalitionindex*}{
    For $i=1,\dots,\ell$, we perform the following procedure. 
        
        By the correctness of our guess and our choice of~$\ind(i)$, we know that $\pi(\wt{C}_i)$ is contained in some coalition~$C_i$ of~$\pi$ with size at least~$\ub_{\ind(i)}$. If $|C_i|=\ub_{\ind(i)}$, then  no modification is needed, and we proceed with the next value for~$i$. 
        If $|C_i| > \ub_{\ind(i)}$, then we modify~$\pi$ as follows. Let $k_i \in [k]$ be the coalition index to which $C_i$ is allocated under~$\pi$; then $\ub_{k_i}=|C_i|$.
        
        First, we create a coalition $C^\star_i$ by taking a  subset of~$C_i$ of size $\ub_{\ind(i)}$ that contains $S(\wt{C}_i)$; we place $C^\star_i$ into the coalition corresponding to~$\ind(i)$. Note that $C^\star_i$ is IR, and moreover, this modification ensures the required condition on~$\pi(\wt{C}_i)$, namely, that it has size~$\ub_{\ind(i)}$.
    
        Second, we add the remaining agents of~$C_i$, that is, all agents in $C_i \setminus C^\star_i$, to the coalition $A_i$ of size~$\ind(i)$ that is originally allocated to coalition index~$\ind(i)$ under~$\pi$. Due to the correctness of our guesses on the ordering $\wt{C}_1,\dots,\wt{C}_\ell$ and our previous steps, all agents of~$\wt{C}_1 \cup \dots \cup \wt{C}_{i-1}$ are contained in coalitions of~$\pi$ with indices at most~$\ind(i)-1$, whereas 
        all agents of~$\wt{C}_{i+1} \cup \dots \cup \wt{C}_\ell$ are contained in coalitions with indices at least~$k_i+1>\ind(i)+1$.
        Thus,
        we know that $A_i$ contains no agents of~$M$. Since $C_i \setminus \wt{C}_i$ contains no  agents of~$M$ either, there are no enmities within the coalition $A_i \cup (C_i \setminus \wt{C}_i)$ which is therefore IR, and moreover, has size exactly~$|C_i|$. We allocate this coalition to the coalition index~$k_i$, to which~$C_i$ was allocated under~$\pi$. 

        Performing these modifications for $i=1,\dots,\ell$, we obtain an IR coalition structure that satisfies the conditions of the claim.
    }
    
    Due to \Cref{clm:locate_coalition_index}, we may restrict our search for a solution in which $\wt{C}_i$ is allocated to the coalition with index~$\ind(i)$. 
    To compute a suitable IR coalition structure, 
    we further need to allocate the agents in~$\agents \setminus M$ to the coalition indices so that (i) %
    each non-lonely agent in~$M$ has a friend in their coalition (recall that we do not know the identity of agents in~$A' \setminus M$), 
    (ii) no lonely agent in~$M$  shares their coalition with an enemy, and (iii) the size constraints are respected.
    
    To this end, we construct a flow network~$N$ as follows, with source~$s$ and target~$t$. Let $I=\{\ind(i):i \in [\ell]\}$ be the indices of coalitions that will hold conflicting agents. 
    For all coalition index~$i \in [\ell]$, we set the \emph{demand} of $\ind(i) \in I$ as $\ub_{\,\ind(i)}-|S(\wt{C}_i)|$, and we set the demand of indices $j \in [k] \setminus I$ simply as $\ub_j$.
    
    First, we create an \emph{agent vertex}~$u_a$ for each  agent~$a \in \agents \setminus M$. Second, for each coalition index~$j \in [k]$, we create a \emph{coalition vertex}~$v_j$. 
    Third, we also define a \emph{guard vertex}~$g_{a'}$ for each $a' \in A' \setminus M$. 
    
    There is an arc~$(s,u_a)$ with capacity~$1$ for each agent~$a \in \agents$, an arc~$(v_j,t)$ for each coalition index $j$ whose capacity is  the demand of~$j$. 
    We also add an arc~$(g_a,t)$ for each agent $a \in A' \setminus M$ with capacity~$1$.
    Now, for each $a \in \agents \setminus M$, we add an arc of capacity~$1$ from agent vertex~$u_a$ to each coalition vertex~$v_j$ where either $j \notin~I$ or $j=\ind(i)$ for some $i \in [\ell]$ such that $a \notin \enemies_c$ for any lonely agent~$c \in \wt{C}_i$. 
    We also add an arc of capacity~$1$ from~$u_a$ to $g_{a'}$ for some~$a' \in A' \setminus M$ if $a$ is \emph{suitable for~$a'$}, meaning that $a$ is considered a friend by all agents~$c$ for which $a'=a_c$. Recall that even though we do not know the identity of~$a'$, we can decide for each $a \in A \setminus M$ whether $a$ is suitable for~$a'$ using the guessed auxiliary graph~$H$---hence, our network is well defined. 
    We finish our algorithm by finding a maximum flow in~$N$; if its value is $|\agents \setminus M|$, we return `yes'; otherwise, we proceed with our next set of guesses. 
    \begin{restatable}[\linkproof{clm:solution_equiv_flow}]{claim}{clmsolutionequivflow}
    \label{clm:solution_equiv_flow}
        Suppose there exists an IR coalition structure~$\pi$ that, for each $i \in [\ell]$, allocates all agents in~$S(\wt{C}_i)$ to coalition index~$\ind(i)$, and the algorithm made correct guesses with respect to~$\pi$. Then 
        there exists a flow of size~$|A \setminus M|$ in the network~$N$.
        Conversely, if~$N$ admits a flow of size~$|A \setminus M|$, then there exists an IR coalition structure~$\pi$.
    \end{restatable}
    \prooftoappendix{clm:solution_equiv_flow}{\clmsolutionequivflow*}{
        We prove the two statements separately.

        \begin{claim}
        \label{clm:solution_gives_flow}
            Suppose there exists an IR coalition structure~$\pi$ that, for each $i \in [\ell]$, allocates all agents in~$S(\wt{C}_i)$ to coalition index~$\ind(i)$, and the algorithm made correct guesses with respect to~$\pi$. Then 
            there exists a flow of size~$|A \setminus M|$ in the network~$N$.    
        \end{claim}
        \begin{claimproof}
        Let $\mu(a)$ denote for each agent~$a \in \agents \setminus M$ the coalition index to which $\pi(a)$ is allocated. Due to the correctness of our guesses, we know that for each agent $a' \in A' \setminus M$ there is an arc from~$u_{a'}$ to the corresponding guard vertex~$g_{a'}$; through this arc, we can define a flow of value~$1$ from~$s$ to~$t$.
        Also, each agent~$a \in \agents \setminus (M \cup A')$ must be allocated to some coalition index~$j$ for which $(u_a,v_j)$ is an arc in the network, because $\pi(a)$ contains no lonely  agent~$c \in M$ for which $a \in \enemies_c$, due to the individual rationality of~$\pi$. We can thus define a flow of value~$1$ flowing through this arc from~$s$ to~$t$. Note that this way, the flow value going through some coalition vertex~$v_j$ is exactly its demand, and hence does not violate the capacity bound on the arc~$(v_j,t)$. Since each agent in~$\agents \setminus M$ adds value~$1$ to the constructed flow, we know that its total value is~$|\agents \setminus M|$.
        \end{claimproof}

        \begin{claim}
        \label{clm:flow_gives_solution}
            Suppose that the network~$N$ admits a flow of size~$|A \setminus M|$. Then there exists an IR coalition structure~$\pi$.
        \end{claim}
        \begin{claimproof}
        Clearly, we may assume that the flow is integral.
        First, observe that the total capacity of all arcs entering~$t$ is the total demand of all coalition indices plus~$|A' \setminus M|$, which equals  $|\agents \setminus \bigcup_{i \in [\ell]}S(\wt{C}_i)|+|A' \setminus M|=|\agents \setminus M|$.
        Therefore, a flow of size~$|\agents \setminus M|$ in~$N$ must saturate all arcs entering~$t$ (or leaving~$s$). 
        
        We construct a coalition structure as follows. 
        For each~$i \in [\ell]$, we put $\wt{C}_i$ into the coalition with index~$\ind(i)$. Then, for each $c \in M'$, we let $b(c)$ denote the agent for which the arc~$(u_{b(c)},g_{a_c})$ carries a flow of value~$1$; such an arc exists, because the arc $(g_{a_c},t)$ is saturated. 
        We add $b(c)$ to the coalition containing~$c$. Since $b(c)$ is suitable for $a_c$, we know that $c$ considers~$b(c)$ a friend. Hence, this ensures that each agent in~$M'$ has a friend in their coalition. 
        Finally, for each coalition index~$j \in [k]$ and for each arc~$(u_a,v_j)$ that is used by the flow, we add  agent~$a$ to the coalition with index~$j$. This ensures that all coalitions satisfy the size constraints, and moreover, by the definition of the network ensures that no agent in $M \setminus M'$  shares their coalition with an enemy, guaranteeing the individual rationality of the obtained coalition structure. 
        \end{claimproof}
        
        \Cref{clm:solution_gives_flow,clm:flow_gives_solution} prove the statement.
        }

    \Cref{clm:locate_coalition_index,clm:solution_equiv_flow} show the correctness of our algorithm.%

    \proofsubparagraph{Running time} 
    Computing~$M$ is in \FPT.
    Guessing the partitioning of~$M$ yields $\vc(G)^{\vc(G)}$ possibilities, while guessing the lonely agents yields $2^{\vc(G)}$ possibilities.
    Guessing the auxiliary graph~$H$ results in at most $\vc(G)^{\vc(G)}$ possibilities, since there are at most~$|M|$ agents in~$A'$, and each agent in~$M'$ is connected to exactly one agent in~$A'$ in~$H$. 
    Finally, guessing the ordering of the sets in~$\pi'$ yields at most $\vc(G)!=\Oh{\vc(G)^{\vc(G)}}$ possibilities.
    Hence, the number of possible guesses is $2^{\vc(G) \log (\vc(G))}$. Since for each fixed set of guesses, the bottleneck is to compute a maximum flow with lower bounds, which can be done in polynomial time with standard techniques,  the running time of the presented algorithm is indeed fixed-parameter tractable with parameter~$\vc(G)$.
\end{proof}

\section{Conclusion}

We studied individual rationality in friend-oriented and enemy-oriented hedonic games under four variants of size constraints, charting the boundary between tractable and intractable cases across a range of natural parameters.

Two messages emerge from our results. First, under friend-oriented preferences, it is the enmity structure alone that governs the complexity: whenever enmities are scarce or well-structured, tractability can often be recovered, whereas rich friendship structures never cause hardness on their own. Second, symmetry acts as a genuine computational resource: several polynomial-time solvable cases for symmetric relations become \NPh as soon as asymmetry is allowed.

Beyond resolving the remaining open cases in \Cref{tab:results}---the most intriguing one, in our view, being whether the size-constrained (SCC) variant is fixed-parameter tractable with respect to the vertex cover number of the relationship graph---we see two natural directions for future research. One is to combine the size constraints studied here with more demanding stability notions, such as Nash or core stability, for which even the existence of unconstrained stable outcomes is non-trivial in our preference domain. Another is optimization: among all feasible individually rational coalition structures, find one that maximizes social welfare.

\begin{credits}
\subsubsection{\ackname} This research was supported in part by the National Science Centre, Poland, grant number UMO-2025/58/A/ST6/00371, and co-funded by the European Union under the project Roboprox (reg. no. CZ.02.01.01/00/22\_008/0004590). Ildikó Schlotter is supported by the Hungarian Academy of Sciences under its Momentum Programme (LP2021-2) and its J\'anos Bolyai Research Scholarship.

\end{credits}

\bibliographystyle{splncs04}
\bibliography{references}

\begin{subappendices}
\renewcommand{\thesection}{\Alph{section}}%
\crefalias{section}{appendix}
\crefalias{subsection}{appendix}
\crefalias{subsubsection}{appendix}

\clearpage
\appendixtext

\bibliographystyleapp{splncs04}
\bibliographyapp{references}

\end{subappendices}

\end{document}